\documentclass[11pt,A4]{article}

\usepackage{graphicx} 
\usepackage[T1]{fontenc}
\usepackage[utf8x]{inputenc}
\usepackage{amsmath, amssymb, amsthm,amscd}
\usepackage{csquotes}
\usepackage[english]{babel}
\usepackage{eurosym}
\usepackage{fullpage}
\usepackage{palatino}
\usepackage{tikz}
\usepackage{verbatim}
\usepackage{fancybox}

\usepackage{subfig}
\usepackage{caption}
\usepackage{enumerate}
\usepackage{epsfig}
\usepackage[active]{srcltx}
\usepackage{amsfonts}
\usepackage[mathlines]{lineno}
\usepackage{xspace}
\usepackage{tikz}

\usetikzlibrary{arrows,shapes,snakes,automata,backgrounds,petri,patterns}

\newtheorem{lemma}{Lemma}

\newtheorem{theorem}{Theorem}
\newtheorem{claim}{Claim}
\newtheorem{corollary}{Corollary}
\newtheorem{conjecture}{Conjecture}

\newcounter{claimb}

\def\claimb{$$\vcenter\bgroup\advance\hsize by -8em\noindent
\refstepcounter{claimb}\ignorespaces\it}        
\makeatletter
\def\endclaimb{\rm\egroup\leqno(\theclaimb)$$\global\@ignoretrue}
\makeatother

\newenvironment{proofclaim}[1][]%
    {\noindent \emph{Proof.} {}{#1}{}}{\hfill
    $\Diamond$\vspace{1em}}
    
\newcommand{\white}[1]{\textcolor{white}{#1}}

\begin{document}

\title{Planar graphs with $\Delta\geq 8$ are ($\Delta+1$)-edge-choosable.\thanks{Work supported by the ANR Grant EGOS (2012-2015) 12 JS02 002 01}}

\author{Marthe Bonamy\\ \normalsize{LIRMM, Universit\'e Montpellier 2}\\ \small{marthe.bonamy@lirmm.fr}}

\maketitle

\begin{abstract}
We consider the problem of \emph{list edge coloring} for planar graphs. Edge coloring is the problem of coloring the edges while ensuring that two edges that are incident receive different colors. A graph is $k$-edge-choosable if for any assignment of $k$ colors to every edge, there is an edge coloring such that the color of every edge belongs to its color assignment. Vizing conjectured in 1965 that every graph is ($\Delta+1$)-edge-choosable. In 1990, Borodin solved the conjecture for planar graphs with maximum degree $\Delta\geq 9$, and asked whether the bound could be lowered to $8$. We prove here that planar graphs with $\Delta\geq 8$ are ($\Delta+1$)-edge-choosable.
\end{abstract}

\section{Introduction}

We consider simple graphs. A \emph{$k$-edge-coloring} of a graph $G$ is a coloring of the edges of $G$ with $k$ colors such that two edges that are incident receive distinct colors. We denote by $\chi'(G)$ the smallest $k$ such that $G$ admits a $k$-edge-coloring. Let $\Delta(G)$ be the maximum degree of $G$. Since incident edges have to receive distinct colors in an edge coloring, every graph $G$ verifies $\chi'(G)\geq \Delta(G)$. A trivial upper-bound on $\chi'(G)$ is $2\Delta(G)-1$, which can can be greatly improved, as follows.

\begin{theorem}[Vizing~\cite{v64}]
Every graph $G$ verifies $\Delta(G) \leq \chi'(G) \leq \Delta(G)+1$.
\end{theorem}

Vizing~\cite{v65} proved that $\chi'(G)=\Delta(G)$ for every planar graph $G$ with $\Delta(G) \geq 8$. He gave examples of planar graphs with $\Delta(G)=4,5$ that are not $\Delta(G)$-edge-colorable, and conjectured that no such graph exists for $\Delta(G)= 6, 7$. This remains open for $\Delta(G)=6$, but the case $\Delta(G)=7$ was solved by Sanders and Zhao~\cite{sz01}, as follows.

\begin{theorem}[Sanders and Zhao~\cite{sz01}]\label{thm:D7}
Every planar graph $G$ with $\Delta(G) \geq 7$ verifies $\chi'(G)=\Delta(G)$.
\end{theorem}

An extension of the problem of edge coloring is the \emph{list edge coloring} problem, defined as follows. For any $L:E\rightarrow \mathcal{P}(\mathbb{N})$ list assignment of colors to the edges of a graph $G=(V,E)$, the graph $G$ is \emph{$L$-edge-colorable} if there exists an edge coloring of $G$ such that the color of every edge $e\in E$ belongs to $L(e)$. A graph $G=(V,E)$ is said to be \emph{list $k$-edge-colorable} (or \emph{$k$-edge-choosable}) if $G$ is $L$-edge-colorable for any list assignment $L$ such that $|L(e)|\geq k$ for any edge $e\in E$. We denote by $\chi'_\ell(G)$ the smallest $k$ such that $G$ is $k$-edge-choosable.\\

One can note that edge coloring is a special case of list edge coloring, where all the lists are equal. Thus $\chi'(G) \leq \chi'_\ell(G)$. This inequality is in fact conjectured to be an equality (see~\cite{jt95} for more information).

\begin{conjecture}[List Coloring Conjecture]\label{conj:LCC}
Every graph $G$ verifies $\chi'(G)=\chi'_{\ell}(G)$.
\end{conjecture}

The conjecture is still widely open. Some partial results were however obtained in the special case of planar graphs: for example, the conjecture is true for planar graphs of maximum degree at least $12$, as follows.

\begin{theorem}[Borodin et al.~\cite{bkw97}]\label{thm:D12}
Every planar graph $G$ with $\Delta(G) \geq 12$ verifies $\chi'_{\ell}(G)=\Delta(G)$.
\end{theorem}

There is still a large gap with the lower bound of $7$ that should hold by Theorem~\ref{thm:D7} if Conjecture~\ref{conj:LCC} were true.

Using Vizing's theorem, the List Coloring Conjecture can be weakened into Conjecture~\ref{conj:LCCw}.

\begin{conjecture}[Vizing~\cite{v76}]\label{conj:LCCw}
Every graph $G$ verifies $\chi'_{\ell}(G) \leq \Delta(G)+1$.
\end{conjecture}

Conjecture~\ref{conj:LCCw} has been actively studied in the case of planar graphs with some restrictions on cycles (see for example~\cite{szhz07, wl01, zw03}), and was settled by Borodin~\cite{b91} for planar graphs of maximum degree at least $9$ (a simpler proof was later found by Cohen and Havet~\cite{ch10}).

\begin{theorem}[Borodin~\cite{b91}]\label{thm:D9}
Every planar graph $G$ with $\Delta(G) \geq 9$ verifies $\chi'_{\ell}(G) \leq \Delta(G)+1$.
\end{theorem}

Here we prove the following theorem.

\begin{theorem}\label{thm:main}
Every planar graph $G$ with $\Delta(G) \leq 8$ verifies $\chi'_{\ell}(G) \leq 9$.
\end{theorem}

This improves Theorem~\ref{thm:D9} and settles Conjecture~\ref{conj:LCCw} for planar graphs of maximum degree $8$.

\begin{corollary}\label{cor:D8}
Every planar graph $G$ with $\Delta(G) \geq 8$ verifies $\chi'_{\ell}(G) \leq \Delta(G)+1$.
\end{corollary}

This answers Problem 5.9 in a survey by Borodin~\cite{b13}. For small values of $\Delta$, Theorem~\ref{thm:main} implies that every planar graph $G$ with $5 \leq \Delta(G) \leq 7$ is also $9$-edge-choosable. To our knowledge, this was not known. It is however known that planar graphs with $\Delta(G) \leq 4$ are $(\Delta(G)+1)$-edge-choosable~\cite{jms99,v76}.

In Sections~\ref{sect:meth} and~\ref{sect:def}, we introduce the method and some terminology. In Sections~\ref{sect:conf} to~\ref{sect:ccl}, we prove Theorem~\ref{thm:main}, with a discharging method.

\section{Method}\label{sect:meth}

The discharging method was introduced in the beginning of the 20$^{th}$ century. It has been used to prove the celebrated Four Color Theorem (\cite{ah77} and \cite{ahk77}).

We prove Theorem~\ref{thm:main} using a discharging method, as follows. A graph is \emph{minimal} for a property if it satisfies this property but none of its proper subgraphs does. The first step is to consider a minimal counter-example $G$ (i.e. a graph $G$ such that $\Delta(G) \leq 8$ and $\chi'_\ell(G) > 9$, whose every proper subgraph is $9$-edge-choosable), and prove it cannot contain some configurations. We assume by contradiction that $G$ contains one of the configurations. We consider a particular subgraph $H$ of $G$. For any list assignment $L$ on the edges of $G$, with $|L(e)|\geq 9$ for every edge $e$, we $L$-edge-color $H$ by minimality. We show how to extend the $L$-edge-coloring of $H$ to $G$, a contradiction.

The second step is to prove that a connected planar graph on at least two vertices with $\Delta \leq 8$ that does not contain any of these configurations does not verify Euler's Formula. To that purpose, we consider a planar embedding of the graph. We assign to each vertex its degree minus six as a weight, and to each face two times its degree minus six. We apply discharging rules to redistribute weights along the graph with conservation of the total weight. As some configurations are forbidden, we can prove that after application of the discharging rules, every vertex and every face has a non-negative final weight. This implies that $\sum_v(d(v)-6)+\sum_f(2d(f)-6) = 2\times |E(G)|- 6\times |V(G)| + 4 \times |E(G)| - 6 \times |F(G)| \geq 0$, a contradiction with Euler's Formula that $|E| - |V| - |F|=-2$. Hence a minimal counter-example cannot exist.

\section{Terminology and notation}\label{sect:def}

In the figures, we draw in black a vertex that has no other neighbor than the ones already represented, in white a vertex that might have other neighbors than the ones represented. When there is a label inside a white vertex, it is an indication on the number of neighbors it has. The label '$i$' means "exactly $i$ neighbors", the label '$i^+$' (resp. '$i^-$') means that it has at least (resp. at most) $i$ neighbors. Note that the white vertices may coincide with other vertices of the figure. The figures are only here to support the text and are not self-sufficient: the embedding can be different from the one represented.

For a given planar embedding, a vertex $v$ is a \emph{weak} neighbor of a vertex $u$ when the two faces adjacent to the edge $(u,v)$ are triangles (see Figure~\ref{fig:weaka}). A vertex $v$ is a \emph{semi-weak} neighbor of a vertex $u$ when one of the two faces adjacent to the edge $(u,v)$ is a triangle and the other is a cycle of length four (see Figure~\ref{fig:weakb}).

\begin{figure}[!h]
\centering
\parbox[b]{2.8in}{
\centering
\begin{tikzpicture}[scale=1.2]
\tikzstyle{whitenode}=[draw,circle,fill=white,minimum size=8pt,inner sep=0pt]
\tikzstyle{blacknode}=[draw,circle,fill=black,minimum size=6pt,inner sep=0pt]
\tikzstyle{invisible}=[draw=white,circle,fill=white,minimum size=6pt,inner sep=0pt]
\tikzstyle{tnode}=[draw,ellipse,fill=white,minimum size=8pt,inner sep=0pt]
\tikzstyle{texte} =[fill=white, text=black]
\draw (0,0) node[whitenode] (u) [label=left:$u$] {}
-- ++(0:1cm) node[whitenode] (v) [label=right:$v$] {};

\draw (u)
-- ++(60:1cm) node[whitenode] (w1) {};
\draw (u)
-- ++(-60:1cm) node[whitenode] (w2) {};

\draw (v) edge node  {} (w1);
\draw (v) edge node  {} (w2);

\end{tikzpicture}
\caption{Vertex $v$ is a weak neighbor of $u$.}
\label{fig:weaka}
}
\qquad
\parbox[b]{3.2in}{
\centering
\begin{tikzpicture}[scale=1.2]
\tikzstyle{whitenode}=[draw,circle,fill=white,minimum size=8pt,inner sep=0pt]
\tikzstyle{blacknode}=[draw,circle,fill=black,minimum size=6pt,inner sep=0pt]
\tikzstyle{tnode}=[draw,ellipse,fill=white,minimum size=8pt,inner sep=0pt]
\tikzstyle{texte} =[fill=white, text=black]
\tikzstyle{invisible}=[draw=white,circle,fill=white,minimum size=8pt,inner sep=0pt]

\draw (0,0) node[whitenode] (u) [label=left:$u$] {}
-- ++(0:1cm) node[whitenode] (v) [label=right:$v$] {};

\draw (u)
-- ++(90:0.8cm) node[whitenode] (w1) {}
-- ++(0:1cm) node[whitenode] (x) {};
\draw (u)
-- ++(-60:1cm) node[whitenode] (w2) {};

\draw (v) edge node  {} (x);
\draw (v) edge node  {} (w2);

\end{tikzpicture}
\caption{Vertex $v$ is a semi-weak neighbor of $u$.}
\label{fig:weakb}
}
\end{figure}

For any vertex $u$, we define special types of weak neighbors of degree $5$ of $u$, as follows. The notation comes from $E$ for ''Eight'' (when $d(u)=8$) and $S$ for ''Seven'' (when $d(u)=7$). The index corresponds to the discharging rules (introduced in Section~\ref{sect:dis}). Consider a weak neighbor $v$ of degree $5$ of $u$.
\begin{itemize}
\item Vertex $v$ is an \emph{$E_2$-neighbor} of $u$ with $d(u)=8$ when one of the two following conditions is verified (see Figure~\ref{fig:Evertex}):
\begin{itemize}
\item There are two vertices $w_1$ and $w_2$ with $d(w_1)=d(w_2)=6$ such that $(u,v,w_1)$ and $(v,w_1,w_2)$ are faces (see Figure~\ref{fig:Evertexa}).
\item There are three vertices $w_1$, $w_2$ and $w_3$ with $d(w_1)=d(w_3)=6$ and $d(w_2)=7$ such that $(u,v,w_1)$, $(v,w_1,w_2)$ and $(u,v,w_3)$ are faces (see Figure~\ref{fig:Evertexb}).
\end{itemize}
\begin{figure}[!h]
\centering
\subfloat[][]{
\begin{tikzpicture}[scale=0.95]
\tikzstyle{whitenode}=[draw,circle,fill=white,minimum size=8pt,inner sep=0pt]
\tikzstyle{blacknode}=[draw,circle,fill=black,minimum size=6pt,inner sep=0pt]
\tikzstyle{invisible}=[draw=white,circle,fill=white,minimum size=6pt,inner sep=0pt]
\tikzstyle{tnode}=[draw,ellipse,fill=white,minimum size=8pt,inner sep=0pt]
\tikzstyle{texte} =[fill=white, text=black]
\draw (0,0) node[whitenode] (u) [label=left:$u$] {\small{$8$}}
-- ++(0:1cm) node[blacknode] (v) [label=right:$v$] {};

\draw (v)
-- ++(180-72:1cm) node[whitenode] (w1) [label=90:$w_1$] {\small{$6$}};
\draw (v)
-- ++(180-2*72:1cm) node[whitenode] (w2) [label=90:$w_2$] {\small{$6$}};
\draw (v)
-- ++(-180+72:1cm) node[whitenode] (w3) [label=-90:\white{$1$}] {};

\draw (v)
-- ++(-180+2*72:1cm) node[whitenode] (w4) {};

\draw (u) edge node  {} (w1);
\draw (w1) edge node  {} (w2);
\draw (w3) edge node  {} (u);

\end{tikzpicture}
\label{fig:Evertexa}
}
\qquad
\subfloat[][]{
\centering
\begin{tikzpicture}[scale=0.95]
\tikzstyle{whitenode}=[draw,circle,fill=white,minimum size=8pt,inner sep=0pt]
\tikzstyle{blacknode}=[draw,circle,fill=black,minimum size=6pt,inner sep=0pt]
\tikzstyle{tnode}=[draw,ellipse,fill=white,minimum size=8pt,inner sep=0pt]
\tikzstyle{texte} =[fill=white, text=black]
\tikzstyle{invisible}=[draw=white,circle,fill=white,minimum size=8pt,inner sep=0pt]

\draw (0,0) node[whitenode] (u) [label=left:$u$] {\small{$8$}}
-- ++(0:1cm) node[blacknode] (v) [label=right:$v$] {};

\draw (v)
-- ++(180-72:1cm) node[whitenode] (w1) [label=90:$w_1$] {\small{$6$}};
\draw (v)
-- ++(180-2*72:1cm) node[whitenode] (w2) [label=90:$w_2$] {\small{$7$}};
\draw (v)
-- ++(-180+72:1cm) node[whitenode] (w3) [label=-90:$w_3$] {\small{$6$}};

\draw (v)
-- ++(-180+2*72:1cm) node[whitenode] (w4) {};

\draw (u) edge node  {} (w1);
\draw (w1) edge node  {} (w2);
\draw (w3) edge node  {} (u);
\end{tikzpicture}
\label{fig:Evertexb}
}
\caption{Vertex $v$ is an $E_2$-neighbor of $u$.}
\label{fig:Evertex}
\end{figure}
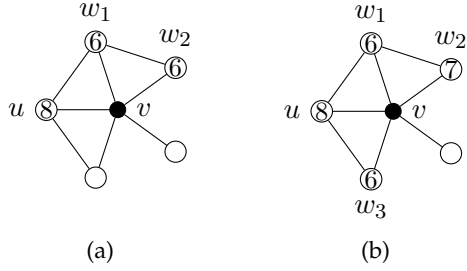

\item Vertex $v$ is an \emph{$E_3$-neighbor} of $u$ with $d(u)=8$ when $v$ is not an $E_2$-neighbor of $u$, and there is a vertex $w$ with $d(w)\leq 7$ such that $(u,v,w)$ is a face (see Figure~\ref{fig:Evertexc}).

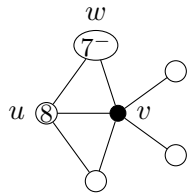
\begin{figure}[!h]
\centering
\begin{tikzpicture}[scale=0.95]
\tikzstyle{whitenode}=[draw,circle,fill=white,minimum size=8pt,inner sep=0pt]
\tikzstyle{blacknode}=[draw,circle,fill=black,minimum size=6pt,inner sep=0pt]
\tikzstyle{tnode}=[draw,ellipse,fill=white,minimum size=8pt,inner sep=0pt]
\tikzstyle{texte} =[fill=white, text=black]
\tikzstyle{invisible}=[draw=white,circle,fill=white,minimum size=8pt,inner sep=0pt]

\draw (0,0) node[whitenode] (u) [label=left:$u$] {\small{$8$}}
-- ++(0:1cm) node[blacknode] (v) [label=right:$v$] {};

\draw (v)
-- ++(180-72:1cm) node[tnode] (w1) [label=90:$w$] {\small{$7^-$}};
\draw (v)
-- ++(180-2*72:1cm) node[whitenode] (w2) {};
\draw (v)
-- ++(-180+72:1cm) node[whitenode] (w3) {};

\draw (v)
-- ++(-180+2*72:1cm) node[whitenode] (w4) {};

\draw (u) edge node  {} (w1);
\draw (w3) edge node  {} (u);
\end{tikzpicture}
\caption{Vertex $v$ is an $E_3$-neighbor of $u$.}
\label{fig:Evertexc}
\end{figure}
\item Vertex $v$ is an \emph{$E_4$-neighbor} of $u$ with $d(u)=8$ when $v$ is not an $E_2$ nor an $E_3$-neighbor of $u$ (see Figure~\ref{fig:Evertexd}). That is, when the third vertices of the two faces containing the edge $(u,v)$ are both of degree $8$.
\begin{figure}[!h]
\centering
\begin{tikzpicture}[scale=0.95]
\tikzstyle{whitenode}=[draw,circle,fill=white,minimum size=8pt,inner sep=0pt]
\tikzstyle{blacknode}=[draw,circle,fill=black,minimum size=6pt,inner sep=0pt]
\tikzstyle{tnode}=[draw,ellipse,fill=white,minimum size=8pt,inner sep=0pt]
\tikzstyle{texte} =[fill=white, text=black]
\tikzstyle{invisible}=[draw=white,circle,fill=white,minimum size=8pt,inner sep=0pt]

\draw (0,0) node[whitenode] (u) [label=left:$u$] {\small{$8$}}
-- ++(0:1cm) node[blacknode] (v) [label=right:$v$] {};

\draw (v)
-- ++(180-72:1cm) node[tnode] (w1) {\small{$8$}};
\draw (v)
-- ++(180-2*72:1cm) node[whitenode] (w2) {};
\draw (v)
-- ++(-180+72:1cm) node[tnode] (w3) {\small{$8$}};

\draw (v)
-- ++(-180+2*72:1cm) node[whitenode] (w4) {};

\draw (u) edge node  {} (w1);
\draw (w3) edge node  {} (u);
\end{tikzpicture}
\caption{Vertex $v$ is an $E_4$-neighbor of $u$.}
\label{fig:Evertexd}
\end{figure}
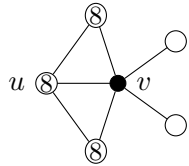
\item Vertex $v$ is an \emph{$S_2$-neighbor} of $u$ with $d(u)=7$ when there are two vertices $w_1$ and $w_2$ with $d(w_1)=d(w_2)=6$ such that $(u,v,w_1)$ and $(u,v,w_2)$ are faces (see Figure~\ref{fig:S2vertexa}).

\begin{figure}[!h]
\centering
\begin{tikzpicture}[scale=0.95]
\tikzstyle{whitenode}=[draw,circle,fill=white,minimum size=8pt,inner sep=0pt]
\tikzstyle{blacknode}=[draw,circle,fill=black,minimum size=6pt,inner sep=0pt]
\tikzstyle{invisible}=[draw=white,circle,fill=white,minimum size=6pt,inner sep=0pt]
\tikzstyle{tnode}=[draw,ellipse,fill=white,minimum size=8pt,inner sep=0pt]
\tikzstyle{texte} =[fill=white, text=black]
\draw (0,0) node[whitenode] (u) [label=left:$u$] {\small{$7$}}
-- ++(0:1cm) node[blacknode] (v) [label=right:$v$] {};

\draw (v)
-- ++(180-72:1cm) node[whitenode] (w1) [label=90:$w_1$] {\small{$6$}};
\draw (v)
-- ++(180-2*72:1cm) node[whitenode] (w2) {};
\draw (v)
-- ++(-180+72:1cm) node[whitenode] (w3) [label=-90:$w_2$] {\small{$6$}};

\draw (v)
-- ++(-180+2*72:1cm) node[whitenode] (w4) {};

\draw (u) edge node  {} (w1);
\draw (w3) edge node  {} (u);

\end{tikzpicture}
\caption{Vertex $v$ is an $S_2$-neighbor of $u$.}
\label{fig:S2vertexa}
\end{figure}
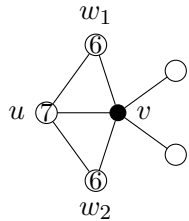

\item Vertex $v$ is an \emph{$S_3$-neighbor} of $u$ with $d(u)=7$ when $v$ is not an $S_2$-neighbor of $u$, and $v$ has four neighbors $w_1$, $w_2$, $w_3$ and $w_4$ such that $(u,v,w_1)$ and $(u,v,w_4)$ are faces, and one of the following two conditions is verified:
\begin{itemize}
\item $(v,w_1,w_2)$, $(v,w_2,w_3)$ and $(v,w_3,w_4)$ are faces, and $d(w_1)=d(w_4)=7$ and $d(w_2)=d(w_3)=6$ (see Figure~\ref{fig:S2vertexb}).
\item $d(w_4)=d(w_2)=6$ and either $d(w_1)=7$ (see Figure~\ref{fig:S2vertexc}) or $d(w_3)=7$ (see Figure~\ref{fig:S2vertexd}). Note that there is no constraint on the order of $w_2$ and $w_3$ in the embedding.
\end{itemize}

\begin{figure}[!h]
\centering
\subfloat[][]{
\centering
\begin{tikzpicture}[scale=0.95]
\tikzstyle{whitenode}=[draw,circle,fill=white,minimum size=8pt,inner sep=0pt]
\tikzstyle{blacknode}=[draw,circle,fill=black,minimum size=6pt,inner sep=0pt]
\tikzstyle{tnode}=[draw,ellipse,fill=white,minimum size=8pt,inner sep=0pt]
\tikzstyle{texte} =[fill=white, text=black]
\tikzstyle{invisible}=[draw=white,circle,fill=white,minimum size=8pt,inner sep=0pt]

\draw (0,0) node[whitenode] (u) [label=left:$u$] {\small{$7$}}
-- ++(0:1cm) node[blacknode] (v) [label=right:$v$] {};

\draw (v)
-- ++(180-72:1cm) node[whitenode] (w1) [label=90:$w_1$] {\small{$7$}};
\draw (v)
-- ++(180-2*72:1cm) node[whitenode] (w2) [label=90:$w_2$] {\small{$6$}};
\draw (v)
-- ++(-180+2*72:1cm) node[whitenode] (w3) [label=-90:$w_3$] {\small{$6$}};

\draw (v)
-- ++(-180+72:1cm) node[whitenode] (w4) [label=-90:$w_4$] {\small{$7$}};

\draw (u) edge node  {} (w1);
\draw (w1) edge node  {} (w2);
\draw (w2) edge node  {} (w3);
\draw (w4) edge node  {} (w3);
\draw (w4) edge node  {} (u);

\end{tikzpicture}
\label{fig:S2vertexb}
}
\qquad
\subfloat[][]{
\centering
\begin{tikzpicture}[scale=0.95]
\tikzstyle{whitenode}=[draw,circle,fill=white,minimum size=8pt,inner sep=0pt]
\tikzstyle{blacknode}=[draw,circle,fill=black,minimum size=6pt,inner sep=0pt]
\tikzstyle{tnode}=[draw,ellipse,fill=white,minimum size=8pt,inner sep=0pt]
\tikzstyle{texte} =[fill=white, text=black]
\tikzstyle{invisible}=[draw=white,circle,fill=white,minimum size=8pt,inner sep=0pt]

\draw (0,0) node[whitenode] (u) [label=left:$u$] {\small{$7$}}
-- ++(0:1cm) node[blacknode] (v) [label=right:$v$] {};

\draw (v)
-- ++(180-72:1cm) node[whitenode] (w1) [label=90:$w_1$] {\small{$7$}};
\draw (v)
-- ++(180-2*72:1cm) node[whitenode] (w2) [label=90:$w_2$] {\small{$6$}};
\draw (v)
-- ++(-180+2*72:1cm) node[tnode] (w3) [label=-90:$w_3$] {\small{}};

\draw (v)
-- ++(-180+72:1cm) node[whitenode] (w4) [label=-90:$w_4$] {\small{$6$}};

\draw (u) edge node  {} (w1);
\draw (w4) edge node  {} (u);

\end{tikzpicture}
\label{fig:S2vertexc}
}
\qquad
\subfloat[][]{
\begin{tikzpicture}[scale=0.95]
\tikzstyle{whitenode}=[draw,circle,fill=white,minimum size=8pt,inner sep=0pt]
\tikzstyle{blacknode}=[draw,circle,fill=black,minimum size=6pt,inner sep=0pt]
\tikzstyle{tnode}=[draw,ellipse,fill=white,minimum size=8pt,inner sep=0pt]
\tikzstyle{texte} =[fill=white, text=black]
\tikzstyle{invisible}=[draw=white,circle,fill=white,minimum size=8pt,inner sep=0pt]

\draw (0,0) node[whitenode] (u) [label=left:$u$] {\small{$7$}}
-- ++(0:1cm) node[blacknode] (v) [label=right:$v$] {};

\draw (v)
-- ++(180-72:1cm) node[whitenode] (w1) [label=90:$w_1$] {\small{}};
\draw (v)
-- ++(180-2*72:1cm) node[whitenode] (w2) [label=90:$w_2$] {\small{$6$}};
\draw (v)
-- ++(-180+2*72:1cm) node[whitenode] (w3) [label=-90:$w_3$] {\small{$7$}};

\draw (v)
-- ++(-180+72:1cm) node[whitenode] (w4) [label=-90:$w_4$] {\small{$6$}};

\draw (u) edge node  {} (w1);
\draw (w4) edge node  {} (u);
\end{tikzpicture}
\label{fig:S2vertexd}
}
\caption{Vertex $v$ is an $S_3$-neighbor of $u$.}
\label{fig:S2vertex}
\end{figure}
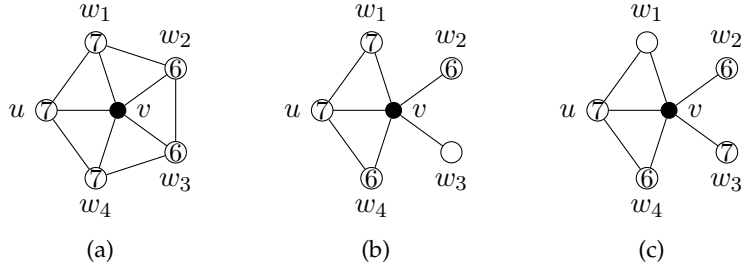
\item Vertex $v$ is an \emph{$S_4$-neighbor} of $u$ with $d(u)=7$ when $v$ is not an $S_2$- nor $S_3$- neighbor of $u$, and either there is a vertex $w$ with $d(w)\leq 7$ such that $(u,v,w)$ is a face (see Figure~\ref{fig:S4vertexa}), or $v$ is adjacent to two vertices $w_1$ and $w_2$ (both distinct from $u$) such that $d(w_1)=6$ and $d(w_2)=7$ (see Figure~\ref{fig:S4vertexb}).
\end{itemize}

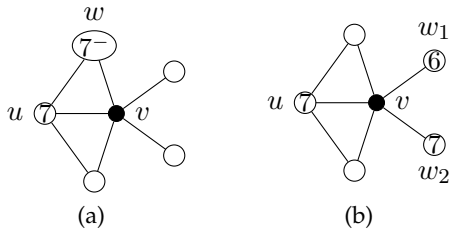
\begin{figure}[!h]
\centering
\subfloat[][]{
\begin{tikzpicture}[scale=0.95]
\tikzstyle{whitenode}=[draw,circle,fill=white,minimum size=8pt,inner sep=0pt]
\tikzstyle{blacknode}=[draw,circle,fill=black,minimum size=6pt,inner sep=0pt]
\tikzstyle{invisible}=[draw=white,circle,fill=white,minimum size=6pt,inner sep=0pt]
\tikzstyle{tnode}=[draw,ellipse,fill=white,minimum size=8pt,inner sep=0pt]
\tikzstyle{texte} =[fill=white, text=black]

\draw (0,0) node[whitenode] (u) [label=left:$u$] {\small{$7$}}
-- ++(0:1cm) node[blacknode] (v) [label=right:$v$] {};

\draw (v)
-- ++(180-72:1cm) node[tnode] (w1) [label=90:$w$] {\small{$7^-$}};
\draw (v)
-- ++(180-2*72:1cm) node[whitenode] (w2) {};
\draw (v)
-- ++(-180+72:1cm) node[whitenode] (w3) {};

\draw (v)
-- ++(-180+2*72:1cm) node[whitenode] (w4) {};

\draw (u) edge node  {} (w1);
\draw (w3) edge node  {} (u);
\end{tikzpicture}
\label{fig:S4vertexa}
}
\qquad
\subfloat[][]{
\begin{tikzpicture}[scale=0.95]
\tikzstyle{whitenode}=[draw,circle,fill=white,minimum size=8pt,inner sep=0pt]
\tikzstyle{blacknode}=[draw,circle,fill=black,minimum size=6pt,inner sep=0pt]
\tikzstyle{tnode}=[draw,ellipse,fill=white,minimum size=8pt,inner sep=0pt]
\tikzstyle{texte} =[fill=white, text=black]
\tikzstyle{invisible}=[draw=white,circle,fill=white,minimum size=8pt,inner sep=0pt]

\draw (0,0) node[whitenode] (u) [label=left:$u$] {\small{$7$}}
-- ++(0:1cm) node[blacknode] (v) [label=right:$v$] {};

\draw (v)
-- ++(180-72:1cm) node[whitenode] (w1) {};
\draw (v)
-- ++(180-2*72:1cm) node[whitenode] (w2) [label=90:$w_1$] {\small{$6$}};
\draw (v)
-- ++(-180+2*72:1cm) node[whitenode] (w3) [label=-90:$w_2$] {\small{$7$}};

\draw (v)
-- ++(-180+72:1cm) node[whitenode] (w4) {};

\draw (u) edge node  {} (w1);
\draw (w4) edge node  {} (u);

\end{tikzpicture}
\label{fig:S4vertexb}
}
\caption{Vertex $v$ is an $S_4$-neighbor of $u$.}
\label{fig:S4vertex}
\end{figure}

Note that if a weak neighbor $v$ of $u$ has none of the previous types, then $u$ and $v$ must verify one of the following four hypotheses:
\begin{itemize}
\item $d(v) \neq 5$
\item $d(u) \not\in \{7;8\}$
\item $d(u)=7$, the third vertex of each of the two triangles adjacent to $(u,v)$ is of degree $8$, and the two other neighbors of $v$ are either both of degree $6$ or both of degree at least $7$.
\end{itemize}

\section{Forbidden Configurations}\label{sect:conf}

We define configurations \textbf{($C_1$)} to \textbf{($C_{11}$)} (see Figure~\ref{fig:forbidden}).
\begin{itemize}
\item \textbf{($C_1$)} is an edge $(u,v)$ with $d(u)+d(v)\leq 10$. 
\item \textbf{($C_2$)} is a cycle $(u,v,w,x)$ such that $d(u)=d(w)=3$. 
\item \textbf{($C_3$)} is a vertex $u$ with $d(u)=8$ that has three neighbors $v_1, v_2$ and $v_3$ such that $v_1$ and $v_2$ are weak neighbors of $u$, with $d(v_1)=d(v_2)=3$ and $d(v_3) \leq 5$.
\item \textbf{($C_4$)} is a vertex $u$ with $d(u)=8$ that has four neighbors $v_1, v_2, v_3$ and $v_4$ such that $v_1$ is a weak neighbor of $u$ and $v_2$ is a semi-weak neighbor of $u$, with $d(v_1)=d(v_2)=3$, $d(v_3) \leq 5$ and $d(v_4) \leq 5$.
\item \textbf{($C_5$)} is a vertex $u$ with $d(u)=8$ that has four weak neighbors $v_1, v_2, v_3$ and $v_4$ with $d(v_1)=3$, $d(v_2)=d(v_3) =4$ and $d(v_4) \leq 5$.
\item \textbf{($C_6$)} is a vertex $u$ with $d(u)=8$ that has five neighbors $v_1,v_2,v_3,v_4$ and $v_5$ such that $v_1$ is a weak neighbor of $u$  with $d(v_1)=3$, $d(v_2) = 4$, $d(v_3) \leq 5$, $d(v_4) \leq 5$ and $d(v_5)\leq 7$.
\item \textbf{($C_7$)} is a vertex $u$ with $d(u)=8$ that has four weak neighbors $v_1, v_2, v_3$ and $v_4$, such that $d(v_1)=3$, vertex $v_2$ is an $E_2$-neighbor of $u$, $d(v_3) \leq 5$ and $d(v_4) \leq 5$.
\item \textbf{($C_8$)} is a vertex $u$ with $d(u)=7$ that has three neighbors $v$, $w$ and $x$ such that $w$ is adjacent to $v$ and $x$, $d(w)=6$, $d(v)=d(x)=5$, and there is a vertex $y$ of degree $6$, distinct from $w$, that is adjacent to $x$.
\item \textbf{($C_9$)} is a vertex $u$ with $d(u)=7$ that has three weak neighbors $v_1,v_2$ and $v_3$ such that $d(v_1)=d(v_2)=4$ and either $v_3$ is an $S_2$, $S_3$ or $S_4$-neighbor, or $d(v_3)=4$.
\item \textbf{($C_{10}$)} is a vertex $u$ with $d(u)=7$ that has three neighbors $v_1, v_2$ and $v_3$ such that $d(v_1)=4$, vertex $v_2$ is an $S_3$-neighbor of $u$ and $d(v_3) \leq 5$.
\item \textbf{($C_{11}$)} is a vertex $u$ with $d(u)=5$ that has three neighbors $v$, $w$ and $x$ such that $w$ is adjacent to $v$ and $x$, and $d(v)=d(w)=d(x)=6$. 
\end{itemize}

\captionsetup[subfloat]{labelformat=empty}
\begin{figure}[!h]
\centering
\subfloat[][\textbf{($C_1$)}]{
\centering
\begin{tikzpicture}[scale=0.95]
\tikzstyle{whitenode}=[draw,circle,fill=white,minimum size=8pt,inner sep=0pt]
\tikzstyle{blacknode}=[draw,circle,fill=black,minimum size=6pt,inner sep=0pt]
\tikzstyle{tnode}=[draw,ellipse,fill=white,minimum size=8pt,inner sep=0pt]
\tikzstyle{texte} =[fill=white, text=black]
\draw (-1,4.2) node[whitenode] (u) [label=90:$u$] {}
-- ++(0:1cm) node[whitenode] (v) [label=90:$v$] {};
\draw (-2,3.4) node[anchor=text] (a) {$d(u)+d(v)\leq 10$};
\end{tikzpicture}
\label{fig:cc1}
}
\qquad
\subfloat[][\textbf{($C_2$)}]{
\centering
\begin{tikzpicture}[scale=0.95]
\tikzstyle{whitenode}=[draw,circle,fill=white,minimum size=8pt,inner sep=0pt]
\tikzstyle{blacknode}=[draw,circle,fill=black,minimum size=6pt,inner sep=0pt]
\tikzstyle{tnode}=[draw,ellipse,fill=white,minimum size=8pt,inner sep=0pt]
\tikzstyle{texte} =[fill=white, text=black] 
\draw (3,5.2) node[whitenode] (u) [label=right:$v$] {}
-- ++(-120:1cm) node[whitenode] (v) [label=90:$u$] {\small{$3$}};
\draw (u)
-- ++(-60:1cm) node[whitenode] (v2) [label=90:$w$] {\small{$3$}};
\draw(v)
-- ++(-60:1cm) node[whitenode] (x) [label=right:$x$] {};
\draw (x) edge node  {} (v2);
\end{tikzpicture}
\label{fig:cc2}
}
\qquad
\subfloat[][\textbf{($C_3$)}]{
\centering
\begin{tikzpicture}[scale=0.95]
\tikzstyle{whitenode}=[draw,circle,fill=white,minimum size=8pt,inner sep=0pt]
\tikzstyle{blacknode}=[draw,circle,fill=black,minimum size=6pt,inner sep=0pt]
\tikzstyle{tnode}=[draw,ellipse,fill=white,minimum size=8pt,inner sep=0pt]
\tikzstyle{texte} =[fill=white, text=black]

\draw (0,0) node[tnode] (u) [label=90:$u$] {\small{$8$}}
-- ++(30:1.155cm) node[tnode] (v3) [label=right:$v_1$ \small{weak}] {\small{$3$}};
\draw (u)
-- ++(0:1cm) node[tnode] (v4) [label=right:$v_2$ \small{weak}] {\small{$3$}};
\draw (u)
-- ++(-30:1.155cm) node[tnode] (v5) [label=right:$v_3$] {\small{$5^-$}};
\end{tikzpicture}
\label{fig:cc3}
}\\
\qquad
\subfloat[][\textbf{($C_4$)}]{
\centering
\begin{tikzpicture}[scale=0.95]
\tikzstyle{whitenode}=[draw,circle,fill=white,minimum size=8pt,inner sep=0pt]
\tikzstyle{blacknode}=[draw,circle,fill=black,minimum size=6pt,inner sep=0pt]
\tikzstyle{tnode}=[draw,ellipse,fill=white,minimum size=8pt,inner sep=0pt]
\tikzstyle{texte} =[fill=white, text=black]

\draw (0,0) node[tnode] (u) [label=90:$u$] {\small{$8$}}
-- ++(45:1.414cm) node[tnode] (v3) [label=right:$v_1$ \small{weak}] {\small{$3$}};
\draw (u)
-- ++(15:1.035cm) node[tnode] (v4) [label=right:$v_2$ \small{semi-weak}] {\small{$3$}};
\draw (u)
-- ++(-15:1.035cm) node[tnode] (v5) [label=right:$v_3$] {\small{$5^-$}};
\draw (u)
-- ++(-45:1.414cm) node[tnode] (v5) [label=right:$v_4$] {\small{$5^-$}};
\end{tikzpicture}
\label{fig:cc4}
}
\qquad
\subfloat[][\textbf{($C_5$)}]{
\centering
\begin{tikzpicture}[scale=0.95]
\tikzstyle{whitenode}=[draw,circle,fill=white,minimum size=8pt,inner sep=0pt]
\tikzstyle{blacknode}=[draw,circle,fill=black,minimum size=6pt,inner sep=0pt]
\tikzstyle{tnode}=[draw,ellipse,fill=white,minimum size=8pt,inner sep=0pt]
\tikzstyle{texte} =[fill=white, text=black]

\draw (0,0) node[tnode] (u) [label=90:$u$] {\small{$8$}}
-- ++(45:1.414cm) node[tnode] (v3) [label=right:$v_1$ \small{weak}] {\small{$3$}};
\draw (u)
-- ++(15:1.035cm) node[tnode] (v4) [label=right:$v_2$ \small{weak}] {\small{$4$}};
\draw (u)
-- ++(-15:1.035cm) node[tnode] (v5) [label=right:$v_3$ \small{weak}] {\small{$4$}};
\draw (u)
-- ++(-45:1.414cm) node[tnode] (v5) [label=right:$v_4$ \small{weak}] {\small{$5^-$}};
\end{tikzpicture}
\label{fig:cc5}
}
\qquad
\subfloat[][\textbf{($C_6$)}]{
\centering
\begin{tikzpicture}[scale=0.95]
\tikzstyle{whitenode}=[draw,circle,fill=white,minimum size=8pt,inner sep=0pt]
\tikzstyle{blacknode}=[draw,circle,fill=black,minimum size=6pt,inner sep=0pt]
\tikzstyle{tnode}=[draw,ellipse,fill=white,minimum size=8pt,inner sep=0pt]
\tikzstyle{texte} =[fill=white, text=black]

\draw (0,0) node[tnode] (u) [label=90:$u$] {\small{$8$}}
-- ++(50:1.555cm) node[tnode] (v3) [label=right:$v_1$ \small{weak}] {\small{$3$}};
\draw (u)
-- ++(30:1.155cm) node[tnode] (v4) [label=right:$v_2$] {\small{$4$}};
\draw (u)
-- ++(0:1cm) node[tnode] (v4) [label=right:$v_3$] {\small{$5^-$}};
\draw (u)
-- ++(-30:1.155cm) node[tnode] (v5) [label=right:$v_4$] {\small{$5^-$}};
\draw (u)
-- ++(-50:1.555cm) node[tnode] (v5) [label=right:$v_5$] {\small{$7^-$}};
\end{tikzpicture}
\label{fig:cc6}
}
\qquad
\subfloat[][\textbf{($C_7$)}]{
\centering
\begin{tikzpicture}[scale=0.95]
\tikzstyle{whitenode}=[draw,circle,fill=white,minimum size=8pt,inner sep=0pt]
\tikzstyle{blacknode}=[draw,circle,fill=black,minimum size=6pt,inner sep=0pt]
\tikzstyle{tnode}=[draw,ellipse,fill=white,minimum size=8pt,inner sep=0pt]
\tikzstyle{texte} =[fill=white, text=black]

\draw (0,0) node[tnode] (u) [label=90:$u$] {\small{$8$}}
-- ++(45:1.414cm) node[tnode] (v3) [label=right:$v_1$ \small{weak}] {\small{$3$}};
\draw (u)
-- ++(15:1.035cm) node[tnode] (v4) [label=right:$v_2$ \small{$E_2$}] {\small{$5$}};
\draw (u)
-- ++(-15:1.035cm) node[tnode] (v5) [label=right:$v_3$ \small{weak}] {\small{$5^-$}};
\draw (u)
-- ++(-45:1.414cm) node[tnode] (v5) [label=right:$v_4$ \small{weak}] {\small{$5^-$}};
\end{tikzpicture}
\label{fig:cc7}
}
\subfloat[][\textbf{($C_8$)}]{
\centering
\begin{tikzpicture}[scale=0.95]
\tikzstyle{whitenode}=[draw,circle,fill=white,minimum size=8pt,inner sep=0pt]
\tikzstyle{blacknode}=[draw,circle,fill=black,minimum size=6pt,inner sep=0pt]
\tikzstyle{tnode}=[draw,ellipse,fill=white,minimum size=8pt,inner sep=0pt]
\tikzstyle{texte} =[fill=white, text=black]

\draw (0,0) node[tnode] (u) [label=90:$u$] {\small{$7$}}
-- ++(45:1cm) node[tnode] (v3) [label=right:$v$] {\small{$5$}};
\draw (u)
-- ++(0:1.414cm) node[tnode] (v4) [label=right:$w$] {\small{$6$}};
\draw (u)
-- ++(-45:1cm) node[tnode] (v5) [label=right:$x$] {\small{$5$}};

\draw (v5)
-- ++(-90:1cm) node[tnode] (v6) [label=right:$y$] {\small{$6$}};

\draw (v4) edge node  {} (v3);
\draw (v4) edge node  {} (v5);
\end{tikzpicture}
\label{fig:cc8}
}
\qquad
\subfloat[][\textbf{($C_9$)}]{
\centering
\begin{tikzpicture}[scale=0.95]
\tikzstyle{whitenode}=[draw,circle,fill=white,minimum size=8pt,inner sep=0pt]
\tikzstyle{blacknode}=[draw,circle,fill=black,minimum size=6pt,inner sep=0pt]
\tikzstyle{tnode}=[draw,ellipse,fill=white,minimum size=8pt,inner sep=0pt]
\tikzstyle{texte} =[fill=white, text=black]

\draw (0,0) node[tnode] (u) [label=90:$u$] {\small{$7$}}
-- ++(30:1.155cm) node[tnode] (v3) [label=right:$v_1$ \small{weak}] {\small{$4$}};
\draw (u)
-- ++(0:1cm) node[tnode] (v4) [label=right:$v_2$ \small{weak}] {\small{$4$}};
\draw (u)
-- ++(-30:1.155cm) node[tnode] (v5) [label=0:$v_3$ \small{weak}][label=-10:\small{$S_2$, $S_3$, $S_4$ or $d(v_3)=4$}] {\small{$5^-$}};
\end{tikzpicture}
\label{fig:cc9}
}\\
\qquad
\subfloat[][\textbf{($C_{10}$)}]{
\centering
\begin{tikzpicture}[scale=0.95]
\tikzstyle{whitenode}=[draw,circle,fill=white,minimum size=8pt,inner sep=0pt]
\tikzstyle{blacknode}=[draw,circle,fill=black,minimum size=6pt,inner sep=0pt]
\tikzstyle{tnode}=[draw,ellipse,fill=white,minimum size=8pt,inner sep=0pt]
\tikzstyle{texte} =[fill=white, text=black]

\draw (0,0) node[tnode] (u) [label=90:$u$] {\small{$7$}}
-- ++(0:1cm) node[tnode] (v3) [label=right:$v_2$ \small{$S_3$}] {\small{$5$}};
\draw (u)
-- ++(30:1.155cm) node[tnode] (v4) [label=right:$v_1$] {\small{$4$}};
\draw (u)
-- ++(-30:1.155cm) node[tnode] (v5) [label=0:$v_3$] {\small{$5^-$}};
\end{tikzpicture}
\label{fig:cc10}
}
\qquad
\subfloat[][\textbf{($C_{11}$)}]{
\centering
\begin{tikzpicture}[scale=0.95]
\tikzstyle{whitenode}=[draw,circle,fill=white,minimum size=8pt,inner sep=0pt]
\tikzstyle{blacknode}=[draw,circle,fill=black,minimum size=6pt,inner sep=0pt]
\tikzstyle{tnode}=[draw,ellipse,fill=white,minimum size=8pt,inner sep=0pt]
\tikzstyle{texte} =[fill=white, text=black]

\draw (0,0) node[tnode] (u) [label=90:$u$] {\small{$5$}}
-- ++(45:1cm) node[tnode] (v3) [label=right:$v$] {\small{$6$}};
\draw (u)
-- ++(0:1.414cm) node[tnode] (v4) [label=right:$w$] {\small{$6$}};
\draw (u)
-- ++(-45:1cm) node[tnode] (v5) [label=0:$x$] {\small{$6$}};

\draw (v4) edge node  {} (v3);
\draw (v4) edge node  {} (v5);
\end{tikzpicture}
\label{fig:cc11}
}
\caption{Forbidden configurations.}
\label{fig:forbidden}
\end{figure}
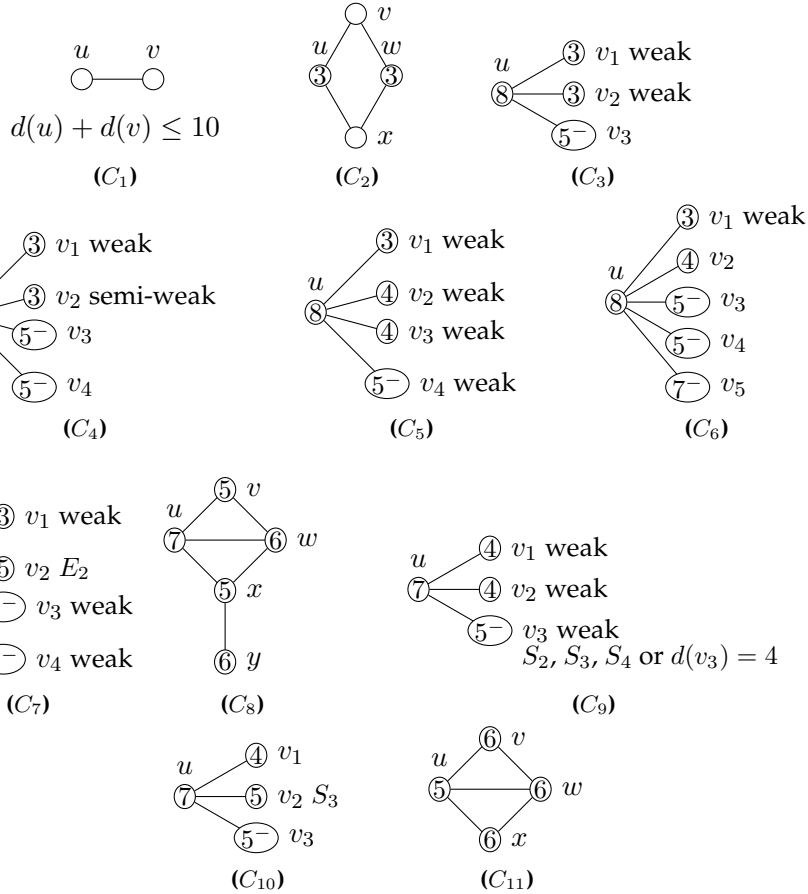
\captionsetup[subfloat]{labelformat=parens}

We first introduce the three following useful lemmas.

\begin{lemma}\label{lem:evencycle}\cite{v76}
Every even cycle $C$ verifies $\chi'_\ell(C)=2$.
\end{lemma}

\begin{lemma}\label{lem:2322}
Let $G$ be the graph with five edges $(a,b,c,d,e)$ such that $(b,c,d,e)$ forms a cycle and $a$ is incident only to $b$ and $e$ (see Figure~\ref{fig:2322}). Let $L:\{a,b,c,d,e\}\rightarrow \mathcal{P}(\mathbb{N})$ a list assignment of at least two colors on every edge, where either $|L(b)|\geq 3$ or $L(b) \neq L(a)$. The graph $G$ is $L$-edge-colorable.
\end{lemma}
\begin{proof}
We consider w.l.o.g. the worst case, i.e. $|L(a)|=|L(c)|=|L(d)|=|L(e)|=2$.
We consider two cases depending on whether $L(c) \cap L(e) = \emptyset$.
\begin{itemize}
\item $L(c) \cap L(e) \neq \emptyset$.\\Then let $\alpha \in L(c) \cap L(e)$. We color $c$ and $e$ in $\alpha$. Since $|L(b)|\geq 3$ or $L(a) \neq L(b)$, we can color $a$ and $b$. We color $d$.
\item $L(c) \cap L(e) = \emptyset$.\\Then let $\alpha \in L(b) \setminus L(a)$. We color $b$ in $\alpha$, and consider two cases depending on whether $\alpha \not\in L(c)$ or $\alpha \not\in L(e)$.
\begin{itemize}
\item $\alpha \not\in L(c)$.\\Then we color successively $e$, $a$, $d$ and $c$. 
\item $\alpha \not\in L(e)$.\\Then we color successively $c$, $d$, $e$ and $a$. 
\end{itemize}
\end{itemize}
\end{proof}

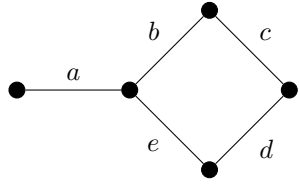
\begin{figure}[!h]
\centering
\begin{tikzpicture}[scale=1.5,auto]
\tikzstyle{whitenode}=[draw,circle,fill=white,minimum size=8pt,inner sep=0pt]
\tikzstyle{blacknode}=[draw,circle,fill=black,minimum size=6pt,inner sep=0pt]
\tikzstyle{tnode}=[draw,ellipse,fill=white,minimum size=8pt,inner sep=0pt]
\tikzstyle{texte} =[fill=white, text=black]
\draw (0,0) node[blacknode] (u) {};
\draw (45:1cm) node[blacknode] (v3) {};
\draw (0:1.414cm) node[blacknode] (v4) {};
\draw (-45:1cm) node[blacknode] (v5) {};
\draw (180:1cm) node[blacknode] (v6) {};
\foreach \source/ \dest /\weight in {u/v3/b, v5/u/e, v3/v4/c, v4/v5/d, v6/u/a}
\draw (\source) edge node[font=\small,pos=0.5] {$\weight$} (\dest);
\end{tikzpicture}
\caption{The graph of Lemma~\ref{lem:2322}.}
\label{fig:2322}
\end{figure}

\begin{lemma}\label{lem:3vertex}
Let $G$ be the star on three edges $(a,b,c)$. Let $L:\{a,b,c\}\rightarrow \mathcal{P}(\mathbb{N})$ a list assignment such that $|L(a)| \geq 2$, $|L(b)| \geq 2$, $|L(c)| \geq 2$. The graph $G$ is $L$-edge-colorable unless $L(a)$, $L(b)$ and $L(c)$ are all equal and of cardinality $2$.
\end{lemma}
\begin{proof}
Assume we do not have $L(a)=L(b)=L(c)$ with $|L(a)| = 2$.  We assume without loss of generality that $|L(a)| \geq 3$ or that $|L(a)|=|L(b)|=|L(c)|=2$ with $L(a) \neq L(b)$ and $L(a) \neq L(c)$. We color $c$ in a color that is not available for $a$ if possible, in an arbitrary color otherwise. We color successively $b$ and $a$.
\end{proof}

\begin{lemma}\label{lem:config}
If $G$ is a minimal planar graph with $\Delta(G) \leq 8$ such that $\chi'_\ell(G)> 9$, then $G$ does not contain any of Configurations \textbf{($C_1$)} to \textbf{($C_{11}$)}.
\end{lemma}

\begin{proof}
Let $L$ be a list assignment on the edges of $G$ with $|L(e)| \geq 9$ for every edge $e$ of $G$. We prove that if $G$ contains any of Configurations \textbf{($C_1$)} to \textbf{($C_{11}$)}, then there is a subgraph $H$ of $G$, that can be $L$-edge-colored by minimality, and whose $L$-edge-coloring is extendable to $G$, a contradiction.\newline

A \emph{constraint} of an edge $e \in E$ is an already colored edge that is incident to $e$. In the following, we denote generically $\hat{e}$ the list of available colors for an edge $e$ at the moment it is used: the list is implicitely modified as incident edges are colored. Proving that the $L$-edge-coloring of $H$ can be extended to $G$ is equivalent to proving that the graph induced by the edges that are not colored yet is $L'$-colorable, where $L'(e)=\hat{e}$ for every edge $e$.
We sometimes \emph{delete} edges. Deleting an edge means that no matter the coloring of the other uncolored edges, there will still be a free color for it (for example, when the edge has more colors available than uncolored incident edges). Thus the deleted edge is implicitely colored after the remaining uncolored edges.\newline

We use the same notations as in the definition of Configurations $(C_1)$ to $(C_{11})$ (see Figure~\ref{fig:forbidden}). 

\begin{claim}\label{claim:C1}
$G$ cannot contain \textbf{($C_1$)}.
\end{claim}
\begin{proof}
Using the minimality of $G$, we color $G \setminus \{(u,v)\}$. Since $d(u)+d(v) \leq 10$, the edge $(u,v)$ has at most $10-2$ constraints. There are $9$ colors, so we can color $(u,v)$, thus extending the coloring to $G$.
\end{proof}

\begin{claim}\label{claim:C2}
$G$ cannot contain \textbf{($C_2$)}.
\end{claim}
\begin{proof}
Using the minimality of $G$, we color $G \setminus \{(u,v),(v,w),(w,x),(x,u)\}$. Since $\Delta(G) \leq 8$ and $d(u)=d(w)=3$, every uncolored edge has at most $8-2+1$ constraints. There are $9$ colors, so every uncolored edge has at least two available colors, and they form a cycle of length four. We can thus apply Lemma~\ref{lem:evencycle} to extend the coloring to $G$.
\end{proof}

\begin{claim}\label{claim:C3}
$G$ cannot contain \textbf{($C_3$)}.
\end{claim}
\begin{proof}
By Claim~\ref{claim:C2}, vertices $v_1$ and $v_2$ have no common neighbor other than $u$. By Claim~\ref{claim:C1}, for $i \in \{1,2,3\}$, vertex $v_i$ is adjacent only to vertices of degree at least $6$. So the $v_i$'s are pairwise non-adjacent. We name the edges according to Figure~\ref{fig:c3}.

\begin{figure}[!h]
\centering
\begin{tikzpicture}[scale=1.5,auto]
\tikzstyle{whitenode}=[draw=black,circle,fill=white,minimum size=8pt,inner sep=0pt]
\tikzstyle{blacknode}=[draw=black,circle,fill=black,minimum size=6pt,inner sep=0pt]
\tikzstyle{tnode}=[draw=black,ellipse,fill=white,minimum size=8pt,inner sep=0pt]
\tikzstyle{texte} =[fill=white, text=black]

\foreach \pos/\name/\angle/\nom in {{(0,0)/u/0/u}, {(22.5+90:2cm)/v1/90/v_1}, {(-22.5-90:2cm)/v2/-90/v_2}}
	\node[blacknode] (\name) [label=\angle:$\nom$] at \pos {};

\foreach \pos/\name in {{(22.5+45:2cm)/x1}, {(22.5+135:2cm)/w1}, {(-22.5-45:2cm)/x2}, {(-22.5-135:2cm)/w2}, {(-22.5:2cm)/v4}}
	\node[whitenode] (\name) at \pos {};

\draw (22.5:2cm) node[tnode] (v3) [label=90:$v_3$] {\small{$5^-$}};

\foreach \i in {1,2}
{\foreach \source/ \dest /\weight in {u/v\i/c_\i, u/w\i/e_\i, u/x\i/f_\i}
\draw[swap] (\source) edge node[font=\small,pos=0.6] {$\weight$} (\dest);}

\foreach \source/ \dest /\weight in {v1/w1/a_1, x1/v1/b_1, w2/v2/a_2, v2/x2/b_2, u/v3/g, u/v4/}
\draw[swap] (\source) edge node[font=\small,pos=0.5] {$\weight$} (\dest);
\end{tikzpicture}
\caption{Notations of Claim~\ref{claim:C3}}
\label{fig:c3}
\end{figure}
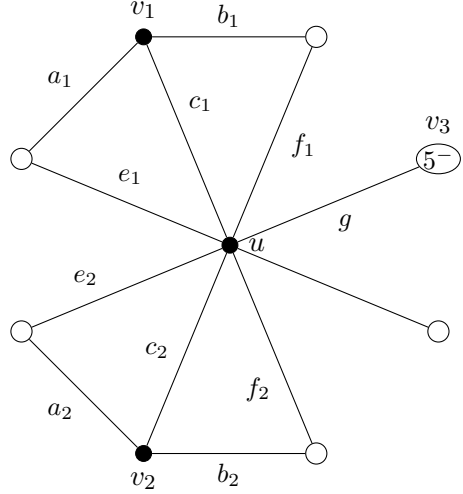

By minimality of $G$, we color $G \setminus \{v_1,v_2\}$. Since there are $9$ colors and every vertex is of degree at most $8$, we have $|\hat{a_1}|,|\hat{a_2}|,|\hat{b_1}|,|\hat{b_2}| \geq 2$ and $|\hat{c_1}|, |\hat{c_2}| \geq 3$. We first prove the following.
\begin{claimb}\label{clb:rec}
If $\hat{a_1}=\hat{b_1}$, $\hat{a_2}=\hat{b_2}$, and $|\hat{a_1}|=|\hat{b_1}|=|\hat{a_2}|=|\hat{b_2}|=2$. Then we can recolor $G \setminus \{v_1,v_2\}$ so that the hypothesis is not verified anymore.
\end{claimb}

\begin{proofclaim}
For $i \in \{1,2\}$, let $\hat{a_i}=\hat{b_i}=\{\alpha_i,\beta_i\}$. For $i \in \{1,2\}$, let $\gamma_i$ and $\delta_i$ be the color of $e_i$ and $f_i$, respectively. Note that $\gamma_i \in L(a_i)$ and $\delta_i \in L(b_i)$ since $|\hat{a_i}|=|\hat{b_i}|=2$. Note that for a given $i \in \{1,2\}$, the colors $\alpha_i, \beta_i, \gamma_i$ and $\delta_i$ are all different. 

We claim that any recoloring of $\{e_1,f_1,e_2,f_2,g\}$ such that the color of at least one of $\{e_1,f_1,$ $e_2,f_2\}$ has been changed breaks the hypothesis of (\ref{clb:rec}). Indeed, assume w.l.o.g. that the color of $e_1$ can be changed while recoloring only edges of $\{e_1,f_1,e_2,f_2,g\}$, and consider such a coloring. We have $\gamma_1 \in \hat{a_1}$ since $\gamma_1 \in L(a_1)$ and the only edge of $\{e_1,f_1,e_2,f_2,g\}$ that is incident to $a_1$ is $e_1$, which is not colored in $\gamma_1$ anymore. We have $\gamma_1 \not\in \hat{b_1}$ since $\gamma_1 \not\in \{\alpha_1,\beta_1,\delta_1\}$ and the only edge of $\{e_1,f_1,e_2,f_2,g\}$ that is incident to $b_1$ is $f_1$, which was colored in $\delta_1$. Thus $\hat{a_1} \neq \hat{b_1}$, and the hypothesis of (\ref{clb:rec}) is broken.

We prove now that there exists such a recoloring. Aside from the constraints derived from $\{e_1,e_2,f_1,f_2,g\}$, each edge $e_i$ or $f_i$ has at most $(8-2)+(8-7)=7$ constraints, and $g$ has at most $(5-1)+(8-7)=5$ constraints. Let $L'$ be the list assignment of the colors available for those edges, when ignoring the constraints derived from $\{e_1,e_2,f_1,f_2,g\}$. Note that $|L'(e_i)|, |L'(f_i)| \geq 2$ and $|L'(g)| \geq 4$. Let us build the directed graph $D$ whose vertex set is $V(D)=\{e_1,e_2,f_1,f_2,g\}$ and where for any two distinct $u,v \in V(D)$, there is an edge from $u$ to $v$ if the color of $u$ belongs to $L'(v)$. We consider two cases depending on whether there is a cycle in $D$.
\begin{itemize}
\item \emph{There is a cycle in $D$}.\\Then we recolor accordingly the edges in $G$ (for any edge from $u$ to $v$ in the cycle, $v$ takes the initial color of $u$, which belongs by definition to $L'(v)$). Since a cycle contains at least two vertices, at least one of $\{e_1,e_2,f_1,f_2\}$ has been recolored.
\item \emph{There is no cycle in $D$}.\\Then some vertex has in-degree $0$. We consider two cases depending on whether some $e_i$ or $f_i$ has in-degree $0$.
\begin{itemize}
\item \emph{Some $e_i$ or $f_i$ has in-degree $0$}.\\Then it can be recolored without conflict (i.e. without recoloring the other vertices of $D$).
\item \emph{Every $e_i$ and $f_i$ has in-degree at least $1$}.\\Then $g$ has in-degree $0$. So $g$ can be recolored without conflict. Since there is no cycle in $D$ and every $e_i$ and $f_i$ is of in-degree at least $1$, there is necessarily an edge from $g$ to some $e_i$ or $f_i$, which can now be recolored without conflict. 
\end{itemize}
\end{itemize}
\end{proofclaim}

By (\ref{clb:rec}), we can assume that we have a coloring of $G \setminus\{v_1,v_2\}$ that does not verify the hypothesis of (\ref{clb:rec}). W.l.o.g., we consider the case where $\hat{a_1}\neq \hat{b_1}$ or $|\hat{a_1}|\geq 3$. We color $a_2, b_2$ and $c_2$. Then $|\hat{c_1}|\geq 2$ and $\hat{a_1}$ and $\hat{b_1}$ have not been modified. So we apply Lemma~\ref{lem:3vertex} to the edges incident to $v_1$.
\end{proof}

\begin{claim}\label{claim:C4}
$G$ cannot contain \textbf{($C_4$)}.
\end{claim}
\begin{proof}
We prove Claim~\ref{claim:C4} similarly as Claim~\ref{claim:C3}.
By Claim~\ref{claim:C2}, vertices $v_1$ and $v_2$ have no common neighbor other than $u$. By Claim~\ref{claim:C1}, for $i \in \{1,2,3,4\}$, vertex $v_i$ is adjacent only to vertices of degree at least $6$. So the $v_i$'s are pairwise non-adjacent. We name the edges according to Figure~\ref{fig:c4}. Note that among the edges named here, the edge $b_2$ is incident only to $a_2$ and $c_2$.

By minimality of $G$, we color $G \setminus \{v_1,v_2\}$. Since there are $9$ colors and every vertex is of degree at most $8$, we have $|\hat{a_1}|,|\hat{a_2}|,|\hat{b_1}|,|\hat{b_2}| \geq 2$ and $|\hat{c_1}|, |\hat{c_2}| \geq 3$. We proceed as for Claim~\ref{claim:C3} and prove the following.
\begin{claimb}\label{clb:rec4}
If $\hat{a_1}=\hat{b_1}$, $\hat{a_2}=\hat{b_2}$, and $|\hat{a_1}|=|\hat{b_1}|=|\hat{a_2}|=|\hat{b_2}|=2$. Then we can recolor $G \setminus \{v_1,v_2\}$ so that the hypothesis is not verified anymore.
\end{claimb}

\begin{proofclaim}
For $i \in \{1,2\}$, let $\hat{a_i}=\hat{b_i}=\{\alpha_i,\beta_i\}$. For $i \in \{1,2\}$, let $\gamma_i$ be the color of $e_i$. Let $\delta_1$ be the color of $f_1$. Note that $\gamma_i \in L(a_i)$ since $|\hat{a_i}|=2$. Similarly, $\delta_1 \in L(b_1)$. Note that for a given $i \in \{1,2\}$, the colors $\alpha_i, \beta_i, \gamma_i$ (and $\delta_1$ if $i=1$) are all different. 

We claim that any recoloring of $\{e_1,f_1,e_2,g_1,g_2\}$ such that the color of at least one of $\{e_1,f_1,e_2\}$ has been changed breaks the hypothesis of (\ref{clb:rec4}). Indeed, assume that the color of $e_1$ can be changed while recoloring only edges of $\{e_1,f_1,e_2,g_1,g_2\}$, and consider such a coloring. (The cases where the color of $f_1$ or $e_2$ can be changed are similar). We have $\gamma_1 \in \hat{a_1}$ since $\gamma_1 \in L(a_1)$ and the only edge of $\{e_1,f_1,e_2,g_1,g_2\}$ that is incident to $a_1$ is $e_1$, which is not colored in $\gamma_1$ anymore. We have $\gamma_1 \not\in \hat{b_1}$ since $\gamma_1 \not\in \{\alpha_1,\beta_1,\delta_1\}$ and the only edge of $\{e_1,f_1,e_2,f_2,g\}$ that is incident to $b_1$ is $f_1$, which was colored in $\delta_1$. Thus $\hat{a_1} \neq \hat{b_1}$.

We prove now that there exists such a recoloring. Aside from the constraints derived from $\{e_1,e_2,f_1,g_1,g_2\}$, each edge $e_1$, $e_2$ and $f_1$ has at most $(8-2)+(8-7)=7$ constraints, and each $g_i$ has at most $4+1=5$ constraints. Let $L'$ be the list assignment of the colors available for those edges, when ignoring the constraints derived from $\{e_1,e_2,f_1,g_1,g_2\}$. Note that $|L'(e_i)|, |L'(f_1)| \geq 2$, and $|L'(g_i)| \geq 4$. We consider w.l.o.g. the worst case, i.e. $|L'(e_1)|=|L'(e_2)|=|L'(f_1)|=2$. Let us build the directed graph $D$ whose vertex set is $V(D)=\{e_1,e_2,f_1,g_1,g_2\}$ and where there is an edge from $u$ to $v$ if the color of $u$ belongs to $L'(v)$. 

First note that if there is an edge from some $g_i$ to some $v \in \{e_1,e_2,f_1\}$, then there are all edges from $\{e_1,e_2,f_1\} \setminus \{v\}$ to $g_i$. Indeed, if vertex $g_i$ has in-degree at most $2$, we recolor $g_i$ and recolor $v$ into the former color of $g_i$. So we assume vertex $g_i$ has in-degree at least $3$. If there is an edge from $v$ to $g_i$, we exchange the colors of $v$ and $g_i$. Thus there are all possible edges from $\{e_1,e_2,f_1\} \setminus \{v\}$ to $g_1$.

If some $e_i$ or $f_i$ has in-degree $0$, se can recolor it without conflict. So we can assume that all of $e_1$, $e_2$ and $f_1$ have in-degree at least $1$. If there is no edge from $\{g_1,g_2\}$ to $\{e_1,e_2,f_1\}$, then there is a directed cycle in $\{e_1,e_2,f_1\}$, and we recolor accordingly the edges in $G$. So there is at least an edge from $\{g_1,g_2\}$ to $\{e_1,e_2,f_1\}$. We consider w.l.o.g. that there is an edge from $g_1$ to $e_1$. By the previous remark, there is an edge from $e_2$ and $f_1$ to $g_1$. Both $e_2$ and $f_1$ have in-degree at least $1$. If there is an edge from $e_2$ to $f_1$ and an edge from $f_1$ to $e_2$, we exchange their colors. So we assume w.l.o.g. that there is an edge from $\{e_1,g_1,g_2\}$ to $e_2$. If there is an edge from $e_1$ to $e_2$, there is a directed cycle on $\{e_1,e_2,g_1\}$, and we recolor accordingly the edges in $G$. If there is an edge from $g_1$ to $e_2$, we exchange the colors of $g_1$ and $e_2$. If there is an edge from $g_2$ to $e_2$, then by the previous remark, there is an edge from $e_1$ to $g_2$. Thus there is a directed cycle on $\{e_1,g_2,e_2,g_1\}$ and we recolor accordingly the edges in $G$. 
\end{proofclaim}

By (\ref{clb:rec4}), we can assume that we have a coloring of $G \setminus\{v_1,v_2\}$ that does not verify the hypothesis of (\ref{clb:rec4}). W.l.o.g., we consider the case where $\hat{a_1}\neq \hat{b_1}$ or $|\hat{a_1}|\geq 3$. We color $a_2, b_2, c_2$ and apply Lemma~\ref{lem:3vertex} to the edges incident to $v_1$.
\end{proof}

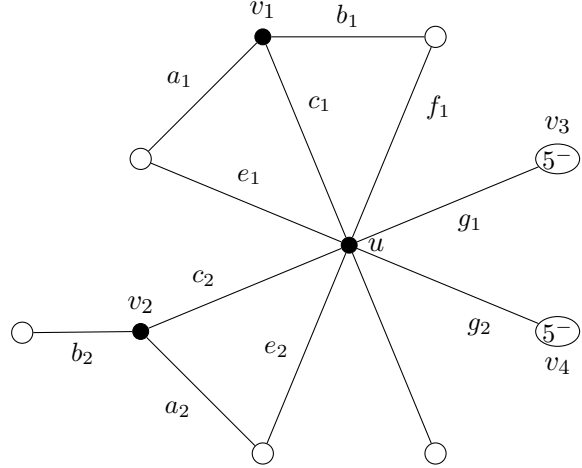
\begin{figure}
\centering
\begin{tikzpicture}[scale=1.5,auto]
\tikzstyle{whitenode}=[draw,circle,fill=white,minimum size=8pt,inner sep=0pt]
\tikzstyle{blacknode}=[draw,circle,fill=black,minimum size=6pt,inner sep=0pt]
\tikzstyle{tnode}=[draw,ellipse,fill=white,minimum size=8pt,inner sep=0pt]
\tikzstyle{texte} =[fill=white, text=black]

\foreach \pos/\name/\angle/\nom in {{(0,0)/u/0/u}, {(22.5+90:2cm)/v1/90/v_1}, {(-22.5-135:2cm)/w2/90/v_2}}
	\node[blacknode] (\name) [label=\angle:$\nom$] at \pos {};

\foreach \pos/\name in {{(22.5+45:2cm)/x1}, {(22.5+135:2cm)/w1}, {(-22.5-45:2cm)/x2}, {(-22.5-90:2cm)/v2}, {(195:3cm)/w3}}
	\node[whitenode] (\name) at \pos {};

\draw (22.5:2cm) node[tnode] (v3) [label=90:$v_3$] {\small{$5^-$}};
\draw (-22.5:2cm) node[tnode] (v4) [label=-90:$v_4$] {\small{$5^-$}};

\foreach \source/ \dest /\weight in {u/v1/c_1, u/w1/e_1, u/v2/e_2,u/w2/c_2}
\draw[swap] (\source) edge node[font=\small,pos=0.6] {$\weight$} (\dest);

\foreach \source/ \dest /\weight in {v1/w1/a_1, x1/v1/b_1, u/v3/g_1, w2/v2/a_2, w3/w2/b_2}
\draw[swap] (\source) edge node[font=\small,pos=0.5] {$\weight$} (\dest);

\draw[swap] (u) edge node[font=\small,pos=0.8] {$g_2$} (v4);
\draw[swap] (u) edge node[font=\small,pos=0.8] {$f_1$} (x1);
\draw[swap] (u) edge node[font=\small] {} (x2);

\end{tikzpicture}
\caption{Notations of Claim~\ref{claim:C4}}
\label{fig:c4}
\end{figure}

\begin{claim}\label{claim:C5}
$G$ cannot contain \textbf{($C_5$)}.
\end{claim}
\begin{proof}
By Claim~\ref{claim:C1}, no two $v_i$ are adjacent. Since every $v_i$ is a weak neighbor of $u$, and $d(u)=8$, the neighborhood of $u$ forms a cycle (see Figure~\ref{fig:c5}). We consider two cases depending on whether there is a vertex $x$ such that $v_2, x$ and $v_3$ appear consecutively around $u$.
\begin{itemize}
\item \emph{There is a vertex $x$ such that $v_2, x$ and $v_3$ appear consecutively around $u$}.\\We consider without loss of generality that the neighbors of $u$ are, clockwise, $v_1$, $w_1$, $v_2$, $w_2$, $v_3$, $w_3$, $v_4$ and $w_4$. We name the edges according to Figure~\ref{fig:c5b}. Note that the edges $l$ and $o$ are distinct. By minimality of $G$, we color $G \setminus \{a,\ldots,r\}$.

Without loss of generality, we consider the worst case, i.e. $|\hat{l}|=|\hat{o}|=|\hat{q}|=|\hat{r}|=2$, $|\hat{b}|=|\hat{d}|=|\hat{f}|=|\hat{h}|=|\hat{i}|=|\hat{j}|=|\hat{k}|=|\hat{m}|=|\hat{n}|=|\hat{p}|=4$, $|\hat{g}|=7$, and $|\hat{a}|=|\hat{c}|=|\hat{e}|=9$. We consider two cases depending on whether $\hat{i}=\hat{j}$ and $\hat{i} \cap \hat{h} \neq \emptyset$.
\begin{itemize}
\item \emph{$\hat{i}\neq \hat{j}$ or $\hat{i} \cap \hat{h} = \emptyset$}.\\If $\hat{i} \neq \hat{j}$, we color $i$ in a color that does not belong to $\hat{j}$. Otherwise $\hat{h} \cap \hat{i} = \emptyset$ and $i$ can be deleted. In any case $|\hat{j}|=4$ and $j$ has exactly $3$ uncolored incident edges, so we can delete it. Then $|\hat{a}|\geq 8$ and $a$ has $7$ uncolored incident edges, so we can delete it. Since $|\hat{b}|+|\hat{l}|>|\hat{k}|$, there exists a color $\alpha \in (\hat{b} \cap \hat{l})\cup ((\hat{b} \cup \hat{l}) \setminus \hat{k})$. Note that $b$ and $l$ are not incident. We color $b$ and $l$ in $\alpha$ if possible, in an arbitrary color otherwise. If $\alpha \in \hat{b}\cap\hat{l}$, then $b$ and $l$ are colored in $\alpha$ and $|\hat{k}|\geq 3$. If $\alpha \in (\hat{b} \cup \hat{l}) \setminus \hat{k}$, then at least one of $b$ and $l$ is colored in $\alpha$ and $|\hat{k}|\geq 3$. So we can delete $k$. Then, successively, $c$, $m$, $e$, $n$, $o$, $p$, $g$, $d$, $f$, $h$, $q$ and $r$ can be deleted.

\item \emph{$\hat{i}=\hat{j}$ and $\hat{i} \cap \hat{h} \neq \emptyset$}.\\Then let $\alpha \in \hat{j} \cap \hat{h}$. Note that $j$ and $h$ are not incident. We color $j$ and $h$ in $\alpha$. Since $i$ (resp. $a$) is incident to both $j$ and $h$, we can successively delete $i$ and $a$. We color successively $r$ and $q$. Without loss of generality, we consider the worst case, i.e. $|\hat{f}|=|\hat{l}|=|\hat{o}|=2$, $|\hat{b}|=|\hat{d}|=|\hat{k}|=|\hat{p}|=3$, $|\hat{g}|=|\hat{m}|=|\hat{n}|=4$, and $|\hat{c}|=|\hat{e}|=8$.
We consider three cases depending on whether $\hat{f} \cap \hat{n} =\emptyset$ and $\hat{p}\setminus \hat{o} \subset \hat{n}$.
\begin{itemize}
\item $\hat{f} \cap \hat{n} \neq \emptyset$.\\Then let $\beta \in \hat{f} \cap \hat{n}$. We color $f$ and $n$ in $\beta$. We delete successively $e$, $p$, $o$, $c$, and $g$. we color $l$. We apply Lemma~\ref{lem:evencycle} on $(b,k,m,d)$.
\item $\hat{p}\setminus \hat{o} \not\subset \hat{n}$.\\Then let $\beta \in \hat{p}\setminus (\hat{o} \cup \hat{n})$. We color $p$ in $\beta$. We color $f$. Since $|\hat{m}|+|\hat{o}|>|\hat{n}|$, there exists $\gamma \in (\hat{m}\cap \hat{o})\cup((\hat{m}\cup \hat{o})\setminus \hat{n})$. If $\gamma \in \hat{d}$ or $\gamma \not\in \hat{m}$, we color $d$ and $o$ in $\gamma$ if possible, in an arbitrary color otherwise. We delete successively $n$, $e$, $c$, $m$, $l$, $k$, $g$ and $b$. If $\gamma \not\in \hat{d}$ and $\gamma \in \hat{m}$, we color $m$ and $o$ in $\gamma$ if possible, in an arbitrary color otherwise. Note that $|\hat{d}|\geq 2$. We delete successively $e$, $c$, $g$, $d$, $b$, $k$ and $l$.
\item \emph{$\hat{f} \cap \hat{n} =\emptyset$ and $\hat{p}\setminus \hat{o} \subset \hat{n}$}.\\Then let $\beta \in \hat{p} \setminus \hat{o}$. We color $p$ in $\beta$. By assumption, $\beta \not\in \hat{f} \cup \hat{o}$. We color $n$ in a color that does not belong to $o$. We delete successively $o$, $e$, $c$ and $g$. We color $l$, and apply Lemma~\ref{lem:2322} on $(f,b,k,m,d)$. 
\end{itemize}
\end{itemize}
\item \emph{There is no vertex $x$ such that $v_2, x$ and $v_3$ appear consecutively around $u$}.\\We consider without loss of generality that the neighbors of $u$ are, clockwise, $v_1$, $w_1$, $v_2$, $w_2$, $v_4$, $w_3$, $v_3$ and $w_4$. We name the edges according to Figure~\ref{fig:c5a}. By minimality of $G$, we color $G \setminus \{a,\ldots,r\}$.

Without loss of generality, we consider the worst case, i.e. $|\hat{l}|=|\hat{n}|=|\hat{o}|=|\hat{q}|=2$, $|\hat{b}|=|\hat{d}|=|\hat{f}|=|\hat{h}|=|\hat{i}|=|\hat{j}|=|\hat{k}|=|\hat{m}|=|\hat{p}|=|\hat{r}|=4$, $|\hat{e}|=7$, and $|\hat{a}|=|\hat{c}|=|\hat{g}|=9$. We consider two cases depending on whether $\hat{i}=\hat{j}$ and $\hat{i} \cap \hat{h} \neq \emptyset$.
\begin{figure}[!h]
\centering
\subfloat[][]{
\begin{tikzpicture}[scale=1.25,auto]
\tikzstyle{whitenode}=[draw,circle,fill=white,minimum size=8pt,inner sep=0pt]
\tikzstyle{blacknode}=[draw,circle,fill=black,minimum size=6pt,inner sep=0pt]
\tikzstyle{tnode}=[draw,ellipse,fill=white,minimum size=8pt,inner sep=0pt]
\tikzstyle{texte} =[fill=white, text=black]
\foreach \pos/\name/\angle/\nom in {{(40:0cm)/u/5/}, {(90:2cm)/v1/90/v_1}, {(0:2cm)/v2/45/v_2}, {(-90:2cm)/v3/(-45)/v_3}}
	\node[blacknode] (\name) [label=\angle: $\nom$] at \pos {};

\foreach \pos/\name/\angle/\nom in {{(45:2cm)/w1/45/w_1}, {(-45:2cm)/w2/-45/w_2}, {(-135:2cm)/w3/-135/w_3}, {(135:2cm)/w4/135/w_4}, {(0:3cm)/x2/90/}, {(-90:3cm)/x3/0/}}
	\node[whitenode] (\name) [label=\angle: $\nom$] at \pos {};

\draw (180:2cm) node[tnode] (v4) [label=180:$v_4$] {\small{$5^-$}};

\foreach \source/ \dest /\weight in {u/v1/a, u/w1/b, u/v2/c, u/w2/d, u/v3/e, u/w3/f, u/v4/g, u/w4/h}
\draw (\source) edge node[font=\small,pos=0.6] {$\weight$} (\dest);

\foreach \source/ \dest /\weight in {w4/v1/i, v1/w1/j, w1/v2/k, x2/v2/l, v2/w2/m, w2/v3/n, x3/v3/o, v3/w3/p, w3/v4/q, v4/w4/r}
\draw (\source) edge node[font=\small,pos=0.5] {$\weight$} (\dest);

\draw (u) edge node[pos=0.2] {$u$} (v2);
\end{tikzpicture}
\label{fig:c5b}
}
\qquad
\subfloat[][]{
\begin{tikzpicture}[scale=1.25,auto]
\tikzstyle{whitenode}=[draw,circle,fill=white,minimum size=8pt,inner sep=0pt]
\tikzstyle{blacknode}=[draw,circle,fill=black,minimum size=6pt,inner sep=0pt]
\tikzstyle{tnode}=[draw,ellipse,fill=white,minimum size=8pt,inner sep=0pt]
\tikzstyle{texte} =[fill=white, text=black]
\foreach \pos/\name/\angle/\nom in {{(40:0cm)/u/5/}, {(90:2cm)/v1/90/v_1}, {(0:2cm)/v2/85/v_2}, {(180:2cm)/v4/95/v_3}}
	\node[blacknode] (\name) [label=\angle: $\nom$] at \pos {};

\foreach \pos/\name/\angle/\nom in {{(45:2cm)/w1/45/w_1}, {(-45:2cm)/w2/-45/w_2}, {(-135:2cm)/w3/-135/w_3}, {(135:2cm)/w4/135/w_4}, {(0:3cm)/x2/90/}, {(180:3cm)/x3/0/}}
	\node[whitenode] (\name) [label=\angle: $\nom$] at \pos {};

\draw (-90:2cm) node[tnode] (v3) [label=-90:$v_4$] {\small{$5^-$}};

\foreach \source/ \dest /\weight in {u/v1/a, u/w1/b, u/v2/c, u/w2/d, u/v3/e, u/w3/f, u/v4/g, u/w4/h}
\draw (\source) edge node[font=\small,pos=0.6] {$\weight$} (\dest);

\foreach \source/ \dest /\weight in {w4/v1/i, v1/w1/j, w1/v2/k, x2/v2/l, v2/w2/m, w2/v3/n, v4/x3/q, v3/w3/o, w3/v4/p, v4/w4/r}
\draw (\source) edge node[font=\small,pos=0.5] {$\weight$} (\dest);

\draw (u) edge node[pos=0.2] {$u$} (v2);
\end{tikzpicture}
\label{fig:c5a}
}

\caption{Notations of Claim~\ref{claim:C5}}
\label{fig:c5}
\end{figure}
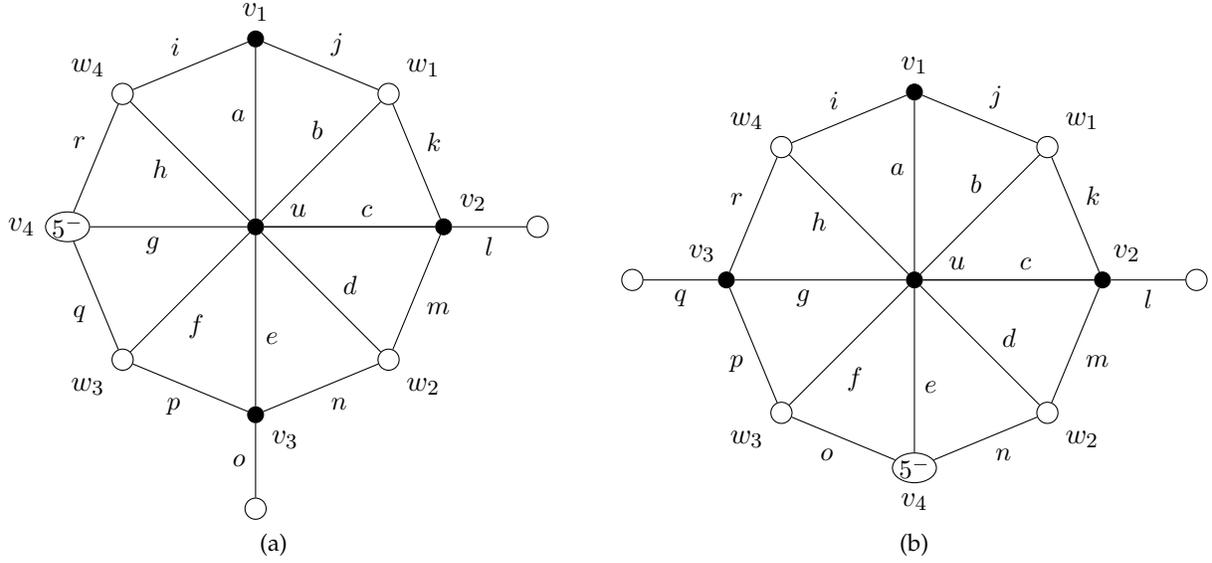
\begin{itemize}
\item \emph{$\hat{i} \neq \hat{j}$ or $\hat{i} \cap \hat{h} = \emptyset$}.\\If $\hat{i} \neq \hat{j}$, we color $i$ in a color that does not belong to $\hat{j}$. Otherwise $\hat{i} \cap \hat{h} = \emptyset$, we can delete $i$. In both cases, we can delete successively $j$ and $a$. 

We consider three cases depending on whether $\hat{q}\cap\hat{h}=\emptyset$ and $\hat{q}\subset \hat{r}$.
\begin{itemize}
\item $\hat{q}\cap\hat{h}\neq\emptyset$.\\Then let $\alpha \in \hat{q}\cap\hat{h}$. We color $q$ and $h$ in $\alpha$. We delete successively $g$, $c$, $e$, $r$, $p$, $k$, $m$, $l$, $b$, $d$, $n$, $p$ and $f$.
\item $\hat{q}\not\subset \hat{r}$.\\Then let $\alpha \in \hat{q} \setminus \hat{r}$. We color $q$ in $\alpha$. Since $|\hat{h}|+|\hat{p}|>|\hat{r}|$, there exists a color $\beta \in (\hat{h}\cap \hat{p})\cup((\hat{h}\cup \hat{p}) \setminus \hat{r})$. We color $h$ and $p$ in $\beta$ if possible, in an arbitrary color otherwise. We delete successively $r$, $g$, $c$, $e$, $k$, $l$, $m$, $b$, $d$, $f$, $n$ and $o$.
\item \emph{$\hat{q}\cap\hat{h}=\emptyset$ and $\hat{q}\subset \hat{r}$}.\\Then let $\alpha \in \hat{q}$. By assumption, $\alpha \in \hat{r} \setminus \hat{h}$. We color $r$ in $\alpha$, and color $q$. Since $|\hat{b}|+|\hat{l}|>|\hat{k}|$, there exists a color $\beta \in (\hat{b}\cap \hat{l})\cup((\hat{b}\cup \hat{l}) \setminus \hat{k})$. We color $b$ and $l$ in $\beta$ if possible, in an arbitrary color otherwise. We delete successively $k$, $c$, $g$, $e$, and $m$. We color $h$ in such a way that afterwards, $|\hat{f}|\geq 3$ or $\hat{f}\neq \hat{p}$. Then we apply Lemma~\ref{lem:2322} on $(p,f,d,n,o)$.
\end{itemize}
\item \emph{$\hat{i}=\hat{j}$ and $\hat{i} \cap \hat{h} \neq \emptyset$}.\\Since $\hat{i}=\hat{j}$, there exists $\alpha \in \hat{j} \cap \hat{h}$, we color $j$ and $h$ in $\alpha$, and delete $i$ and $a$. When we say that we color $q|r$ in a color $\alpha$, it means that we color $q$ in $\alpha$ if possible, otherwise we color $r$ in $\alpha$.

Let $C=\hat{c}$ and $G=\hat{g}$. If $\hat{q}\cup\hat{r} \not\subset \hat{g}$, we consider $\alpha \in (\hat{q}\cup\hat{r}) \setminus \hat{g}$, and color $q|r$ in $\alpha$. Assume that $\hat{q}\cup\hat{r}\subset \hat{g}$. Note that $|((\hat{q}\cup\hat{r})\cap\hat{c})\cup(\hat{c}\setminus\hat{g})|\geq |\hat{q}\cup\hat{r}|\geq 3$, and that $|\hat{l}|=2$. We consider $\alpha \in (((\hat{q}\cup\hat{r})\cap\hat{c})\cup(\hat{c}\setminus\hat{g}))\setminus \hat{l}$. We color $q|r$ in $\alpha$ if possible, in an arbitrary color otherwise. 

Note that since $q$ and $r$ have the same incidencies in the resulting graph, and since $|\hat{r}|\geq|\hat{q}|-1$, the identity of the edge that is colored has no impact, and we can consider w.l.o.g. that $q$ is colored and $r$ remains uncolored. We remove color $\alpha$ from $\hat{k}$ and $\hat{m}$. We consider w.l.o.g. the worst case, i.e. $|\hat{k}|=|\hat{l}|=|\hat{n}|=|\hat{o}|=|\hat{r}|=2$, $|\hat{b}|=|\hat{d}|=|\hat{f}|=|\hat{m}|=|\hat{p}|=3$, $|\hat{e}|=6$, $|\hat{g}|=7$ and $|\hat{c}|=8$.

We consider two cases depending on whether $\hat{k}=\hat{l}$.
\begin{itemize}
\item $\hat{k}=\hat{l}$. Then we color $m$ in a color that does not belong to $\hat{l}$. We color successively $n$, $o$, $d$, $f$, $b$, $k$ and $l$.
\item $\hat{k}\neq\hat{l}$. Then we color $l$ in a color that does not belong to $\hat{k}$. If $\hat{m}=\hat{n}$, then we color $d$ in a color that does not belong to $\hat{m}$, and apply Lemma~\ref{lem:evencycle} on $(b,k,m,n,o,f)$. If $\hat{m}\neq \hat{n}$, then we color $m$ in a color that does not belong to $\hat{n}$, we color $k$ and we apply Lemma~\ref{lem:2322} on $(b,f,o,n,d)$.
\end{itemize}

We then color $p$, $q$ and $e$. We claim that $\hat{c} \neq \hat{g}$ if $|\hat{c}|=|\hat{g}|=1$. Indeed, assume $|\hat{c}|=|\hat{g}|=1$. Then, all the edges incident to $g$ are colored differently, and their colors belong to $G$. We consider two cases depending on whether $q$ is colored in $\alpha$.
\begin{itemize}
\item \emph{Edge $q$ is colored in $\alpha$}. Then $\alpha \in G$, which implies $\alpha \in C$ by choice of $\alpha$. Since the edges incident to $g$ are all colored differently and $q$ is colored in $\alpha$, none of $\{b,d,e,f,h\}$ is colored in $\alpha$. By construction, none of $\{k,m\}$ is colored in $\alpha$. By choice of $\alpha$, $l$ is not colored in $\alpha$. Thus $\alpha \in \hat{c}$ and $\alpha \not\in \hat{g}$, so $\hat{c}\neq\hat{g}$. 
\item \emph{Edge $q$ is not colored in $\alpha$}. Then, by choice of $\alpha$, we have $\alpha \in C \setminus G$. Since the colors of the edges incident to $g$ all belong to $G$, none of $\{b,d,e,f,h\}$ is colored in $\alpha$. By construction, none of $\{k,m\}$ is colored in $\alpha$. By choice of $\alpha$, $l$ is not colored in $\alpha$. Thus $\alpha \in \hat{c}$ and $\alpha \not\in \hat{g}$, so $\hat{c}\neq\hat{g}$.
\end{itemize}

Note that $|\hat{c}|\geq 1$ and $|\hat{g}|\geq 1$. If $|\hat{c}|=|\hat{g}|=1$, then $\hat{c} \neq \hat{g}$, so we color $c$ and $g$ independently. If not, assume w.l.o.g. that $|\hat{c}|\geq 2$, and color successively $g$ and $c$. 
\end{itemize}
\end{itemize}
\end{proof}

\begin{claim}\label{claim:C6}
$G$ cannot contain \textbf{($C_6$)}.
\end{claim}
\begin{proof}
We prove Claim~\ref{claim:C6} similarly as Claim~\ref{claim:C3}.
By Claim~\ref{claim:C1}, for $i \in \{1,2,3,4,5\}$, vertex $v_i$ is adjacent only to vertices of degree at least $11-d(v_i)$. We name the edges according to Figure~\ref{fig:c6}.\begin{figure}[!h]
\centering
\begin{tikzpicture}[scale=1.25,auto]
\tikzstyle{whitenode}=[draw,circle,fill=white,minimum size=8pt,inner sep=0pt]
\tikzstyle{blacknode}=[draw,circle,fill=black,minimum size=6pt,inner sep=0pt]
\tikzstyle{tnode}=[draw,ellipse,fill=white,minimum size=8pt,inner sep=0pt]
\tikzstyle{texte} =[fill=white, text=black]
\foreach \pos/\name/\angle/\nom in {{(40:0cm)/u/5/}, {(90:2cm)/v1/90/v_1}}
	\node[blacknode] (\name) [label=\angle: $\nom$] at \pos {};

\foreach \pos/\name in {{(45:2cm)/w1}, {(135:2cm)/w4}, {(180:2cm)/w/}}
	\node[whitenode] (\name) at \pos {};

\draw (0:2cm) node[tnode] (v2) [label=-90:$v_2$] {\small{$4^-$}};
\draw (-45:2cm) node[tnode] (v3) [label=-90:$v_3$] {\small{$5^-$}};
\draw (-90:2cm) node[tnode] (v4) [label=-90:$v_4$] {\small{$5^-$}};
\draw (-135:2cm) node[tnode] (v5) [label=-90:$v_5$] {\small{$7^-$}};

\foreach \source/ \dest /\weight in {u/v1/c, u/w1/f, u/v2/g_1, u/v3/g_2, u/v4/g_3, u/w4/e, u/v5/g_4, u/w/}
\draw (\source) edge node[font=\small,pos=0.6] {$\weight$} (\dest);

\foreach \source/ \dest /\weight in {w4/v1/a, v1/w1/b}
\draw (\source) edge node[font=\small,pos=0.5] {$\weight$} (\dest);

\draw (u) edge node[pos=0.2] {$u$} (v2);
\end{tikzpicture}
\caption{Notations of Claim~\ref{claim:C6}}
\label{fig:c6}
\end{figure}
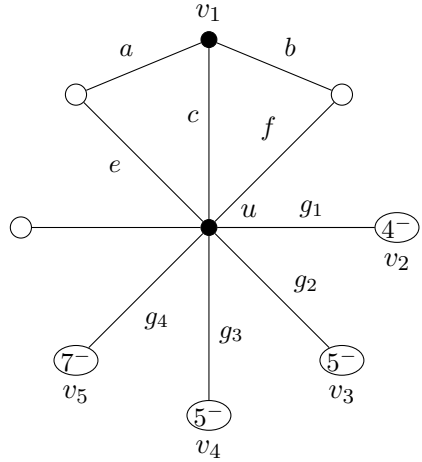 By minimality of $G$, we color $G \setminus \{v_1\}$. Since there are $9$ colors and every vertex is of degree at most $8$, we have $|\hat{a}|,|\hat{b}|, |\hat{c}| \geq 2$. We proceed as for Claim~\ref{claim:C3} and prove the following.
\begin{claimb}\label{clb:rec6}
If $\hat{a}=\hat{b}$ and $|\hat{a}|=|\hat{b}|=2$. Then we can recolor $G \setminus \{v_1\}$ so that the hypothesis is not verified anymore.
\end{claimb}

\begin{proofclaim}
Let $\hat{a}=\hat{b}=\{\alpha,\beta\}$. Let $\gamma$ be the color of $e$ and $\delta$ the color of $f$. Note that $\gamma \in L(a)$ and $\delta \in L(b)$ since $|\hat{a}|=|\hat{b}|=2$. Note also that $\alpha, \beta, \gamma$ and $\delta$ are all different.

We claim that any recoloring of $\{e,f,g_1,g_2,g_3,g_4\}$ such that the color of at least one of ${e,f}$ has been changed breaks the hypothesis of (\ref{clb:rec6}). Indeed, assume w.l.o.g. that the color of $e$ can be changed while recoloring only edges of $\{e,f,g_1,g_2,g_3,g_4\}$, and consider such a coloring. We have $\gamma \in \hat{a}$ since $\gamma \in L(a)$ and the only edge of $\{e,f,g_1,g_2,g_3,g_4\}$ that is incident to $a$ is $e$, which is not colored in $\gamma$ anymore. We have $\gamma \not\in \hat{b}$ since $\gamma \not\in \{\alpha,\beta,\delta\}$ and the only edge of $\{e,f,g_1,g_2,g_3,g_4\}$ that is incident to $b$ is $f$, which was colored in $\delta$. Thus $\hat{a}\neq \hat{b}$.

We prove that there exists such a recoloring. Aside from the constraints derived from $\{e,f,$ $g_1,g_2,$ $g_3,g_4\}$, both $e$ and $f$ have at most $7$ constraints, edge $g_1$ has at most $4$ constraints, edges $g_2$ and $g_3$ have at most $5$ constraints, and $g_4$ has at most $7$ constraints. Let $L'$ be the list assignment of the colors available for those edges, when ignoring the constraints derived from $\{e,f,g_1,g_2,g_3,g_4\}$. Note that $|L'(e)|, |L'(f)|, |L'(g_4)| \geq 2$, $|L'(g_2)|, |L'(g_3)| \geq 4$ and $|L'(g_1)|\geq 5$. Let us build the directed graph $D$ whose vertex set is $V(D)=\{e,f,g_1,g_2,g_3,g_4\}$ and where there is an edge from $u$ to $v$ if the color of $u$ belongs to $L'(v)$. Let $D_1$ be the graph obtained from $D$ by removing any vertex $v$ such that there is no directed path from $v$ to $e$. Let $D_2$ be the graph obtained from $D$ by removing any vertex $v$ such that there is no directed path from $v$ to $f$. If $e \in D_2$ and $f \in D_1$, then there is a directed path from $e$ to $f$ and a directed path from $f$ to $e$. So there exists a directed cycle that contains $e$, which we recolor accordingly. So we can assume that $e \not\in D_2$ or $f \not\in D_1$. We consider w.l.o.g. the case $f \not\in D_1$.
We consider four cases depending on the structure of $D_1$.
\begin{itemize}
\item $V(D_1)=\{e\}$. Then we recolor $e$ without conflict.
\item \emph{$|V(D_1)|\geq 2$, and some vertex $v \neq e$ has in-degree at most $L'(v)-2$}. Then we recolor $v$, and recolor accordingly the path from $v$ to $e$.
\item \emph{$|V(D_1)|\geq 2$, and there is an edge from $e$ to a vertex $v$}. Then by definition of $D_1$, there is a directed cycle that contains $e$, which we recolor accordingly.
\item \emph{$|V(D_1)|\geq 2$, every vertex $v \neq e$ has in-degree at least $L'(v)-1$, and $e$ has out-degree $0$}. Since $f \not\in D_1$, we have $\{g_1,g_2,g_3,g_4\}\cap D_1 \neq \emptyset$. Let $j$ be the minimum $i$ such that $g_i \in D_1$. Vertex $g_j$ has in-degree at least $L'(g_j)-1 \geq |V(D) \setminus \{f,g_1,\ldots,g_{j}\}|$, and there is no edge from $e$ to $g_j$, a contradiction. 
\end{itemize}
\end{proofclaim}

By (\ref{clb:rec6}), we can assume that we have a coloring of $G \setminus\{v_1\}$ that does not verify the hypothesis of (\ref{clb:rec6}). We apply Lemma~\ref{lem:3vertex} to the edges incident to $v_1$.
\end{proof}

\begin{claim}\label{claim:C7}
$G$ cannot contain \textbf{($C_7$)}.
\end{claim}
\begin{proof}
By Claim~\ref{claim:C1}, no two $v_i$ are adjacent, nor is $v_1$ adjacent to a vertex of degree at most $7$. Since every $v_i$ is a weak neighbor of $u$, and $d(u)=8$, the neighborhood of $u$ forms a cycle (see Figure~\ref{fig:c7}). We consider two cases depending on whether there is a vertex $x$ such that $v_2, x$ and $v_3$ appear consecutively around $u$.
\begin{itemize}
\item \emph{There is a vertex $x$ such that $v_3, x$ and $v_4$ appear consecutively around $u$}.\\
W.l.o.g. the neighbors of $u$ are, clockwise, $v_1$, $x_1$, $v_2$, $x_2$, $v_3$, $x_3$, $v_4$, $x_4$. Since $v_2$ is an $E_2$-neighbor of $u$ and $d(x_1)=8$, we have $d(x_2)=6$ and there is a vertex $y$ of degree $6$ such that $(x_2,v_2,y)$ is a face. We name the edges according to Figure~\ref{fig:c7a}. By minimality, we color $G\setminus\{a,\ldots,s\}$. Without loss of generality, we consider the worst case, i.e. $|\hat{n}|=|\hat{o}|=|\hat{p}|=|\hat{q}|=2$, $|\hat{s}|=3$, $|\hat{b}|=|\hat{f}|=|\hat{h}|=|\hat{i}|=|\hat{j}|=|\hat{k}|=4$, $|\hat{m}|=|\hat{r}|=5$, $|\hat{d}|=|\hat{e}|=|\hat{g}|=|\hat{l}|=7$, and $|\hat{a}|=|\hat{c}|=9$. Note that the edges $k$ and $s$ are not incident. Since $|\hat{k}|+|\hat{s}|>|\hat{r}|$, there exists $\alpha \in (\hat{k} \cap \hat{s})\cup((\hat{k} \cup \hat{s})\setminus \hat{r})$. We color $k$ and $s$ in $\alpha$ if possible, in an arbitrary color otherwise. We can delete successively $r$, $l$ and $m$. We color $q$. Note that the edges $p$ and $j$ are not incident. Since $|\hat{p}|+|\hat{j}|>|\hat{i}|$, there exists $\beta \in (\hat{p} \cap \hat{j})\cup((\hat{p} \cup \hat{j})\setminus \hat{i})$. Thus we color $p$ and $j$ in $\beta$ if possible, in an arbitrary color otherwise. We delete successively $i$, $a$, $c$, $e$, $g$, $d$, $h$, $b$, $f$, $n$ and $o$.

\item \emph{There is no vertex $x$ such that $v_2, x$ and $v_3$ appear consecutively around $u$}.\\
W.l.o.g. the neighbors of $u$ are, clockwise, $v_1$, $x_1$, $v_3$, $x_2$, $v_2$, $x_3$, $v_4$, $x_4$, with $d(x_2) \geq d(x_3)$. We consider two cases depending on whether $d(x_2)=6$.
\begin{itemize}
\item $d(x_2)=6$.\\W.l.o.g., since $v_2$ is an $E_2$-vertex, there is a vertex $y$ of degree $6$ or $7$ such that $(y,v_2,x_3)$ is a face. We name the edges according to Figure~\ref{fig:c7b}. By minimality, we color $G\setminus\{a,\ldots,s\}$. W.l.o.g., we consider the worst case, i.e. $|\hat{k}|=|\hat{p}|=|\hat{q}|=|\hat{s}|=2$, $|\hat{b}|=|\hat{h}|=|\hat{i}|=|\hat{j}|=|\hat{l}|=|\hat{r}|=4$, $|\hat{o}|=5$, $|\hat{d}|=|\hat{m}|=6$, $|\hat{c}|=|\hat{f}|=|\hat{g}|=|\hat{n}|=7$, and $|\hat{a}|=|\hat{e}|=9$.

Note that the edges $r$ and $l$ cannot be incident. Since $|\hat{r}|+|\hat{l}|>|\hat{m}|$, there exists $\alpha \in (\hat{r} \cap \hat{l})\cup((\hat{r} \cup \hat{l})\setminus \hat{m})$. We color $r$ and $l$ in $\alpha$ if possible, in an arbitrary color otherwise. We can delete successively $m$, $n$ and $o$. We color $q$, $s$ and $k$ successively. Note that $p$ and $j$ cannot be incident. Since $|\hat{p}|+|\hat{j}|>|\hat{i}|$, there exists $\beta \in (\hat{p} \cap \hat{j})\cup((\hat{p} \cup \hat{j})\setminus \hat{i})$. Thus we color $p$ and $j$ in $\beta$ if possible, in an arbitrary color otherwise. We delete successively $i$, $a$, $e$, $g$, $f$, $c$, $d$, $h$ and $b$.

\item $d(x_2)\geq 7$.\\Since $v_2$ is an $E_2$-vertex, we have $d(x_3)=6$ and there is a vertex $y$ of degree $6$ such that $(y,v_2,x_3)$ is a triangle. We name the edges according to Figure~\ref{fig:c7c}. By minimality, we color $G\setminus\{a,\ldots,s\}$. Without loss of generality, we consider the worst case, i.e. $|\hat{k}|=|\hat{l}|=|\hat{p}|=|\hat{q}|=2$, $|\hat{s}|=3$, $|\hat{b}|=|\hat{d}|=|\hat{h}|=|\hat{i}|=|\hat{j}|=|\hat{m}|=4$, $|\hat{o}|=|\hat{r}|=5$, $|\hat{c}|=|\hat{f}|=|\hat{g}|=|\hat{n}|=7$, and $|\hat{a}|=|\hat{e}|=9$.
\begin{figure}[!h]
\centering
\subfloat[][]{
\begin{tikzpicture}[scale=1.25,auto]
\tikzstyle{whitenode}=[draw,circle,fill=white,minimum size=8pt,inner sep=0pt]
\tikzstyle{blacknode}=[draw,circle,fill=black,minimum size=6pt,inner sep=0pt]
\tikzstyle{tnode}=[draw,ellipse,fill=white,minimum size=8pt,inner sep=0pt]
\tikzstyle{texte} =[fill=white, text=black]
\foreach \pos/\name/\angle/\nom in {{(40:0cm)/u/5/}, {(90:2cm)/v1/90/v_1}, {(0:2cm)/v2/0/v_2}}
	\node[blacknode] (\name) [label=\angle: $\nom$] at \pos {};

\foreach \pos/\name/\angle/\nom in {{(45:2cm)/x1/45/x_1}, {(135:2cm)/x4/135/x_4}, {(-135:2cm)/x3/-135/x_3}}
	\node[whitenode] (\name) [label=\angle: $\nom$] at \pos {};

\draw (v2)
 ++(45:1.414cm) node[whitenode] (z) {};
 \draw (v2)
 ++(-45:1.414cm) node[whitenode] (y) [label=-45:$y$] {\small{$6$}};

\draw (-45:2cm) node[tnode] (x2) [label=-45:$x_2$] {\small{$6$}};
\draw (-90:2cm) node[tnode] (v3) [label=-90:$v_3$] {\small{$5$}};
\draw (180:2cm) node[tnode] (v4) [label=180:$v_4$] {\small{$5^-$}};

\foreach \source/ \dest /\weight in {u/v1/a, u/x1/b, u/v2/c, u/x2/d, u/v3/e, u/x3/f, u/v4/g, u/x4/h, v2/z/q, v2/y/r}
\draw (\source) edge node[font=\small,pos=0.6] {$\weight$} (\dest);

\foreach \source/ \dest /\weight in {x4/v1/i, v1/x1/j, x1/v2/k, v2/x2/l, x2/v3/m, v3/x3/n, x3/v4/o, v4/x4/p, y/x2/s}
\draw (\source) edge node[font=\small,pos=0.5] {$\weight$} (\dest);

\draw (u) edge node[pos=0.2] {$u$} (v2);
\end{tikzpicture}
\label{fig:c7a}
}
\qquad
\subfloat[][]{
\begin{tikzpicture}[scale=1.25,auto]
\tikzstyle{whitenode}=[draw,circle,fill=white,minimum size=8pt,inner sep=0pt]
\tikzstyle{blacknode}=[draw,circle,fill=black,minimum size=6pt,inner sep=0pt]
\tikzstyle{tnode}=[draw,ellipse,fill=white,minimum size=8pt,inner sep=0pt]
\tikzstyle{texte} =[fill=white, text=black]
\foreach \pos/\name/\angle/\nom in {{(40:0cm)/u/5/}, {(90:2cm)/v1/90/v_1}, {(-90:2cm)/v3/-90/v_2}}
	\node[blacknode] (\name) [label=\angle: $\nom$] at \pos {};

\foreach \pos/\name/\angle/\nom in {{(45:2cm)/x1/45/x_1}, {(135:2cm)/x4/135/x_4}}
	\node[whitenode] (\name) [label=\angle: $\nom$] at \pos {};

\draw (v3)
 ++(-45:1.414cm) node[whitenode] (z) {};
 \draw (v3)
 ++(-135:1.414cm) node[tnode] (y) [label=-135:$y$] {\small{$7^-$}};

\draw (-135:2cm) node[tnode] (x3) [label=-135:$x_3$] {\small{$6$}};
\draw (-45:2cm) node[tnode] (x2) [label=-45:$x_2$] {\small{$6$}};
\draw (0:2cm) node[tnode] (v2) [label=0:$v_3$] {\small{$5$}};
\draw (180:2cm) node[tnode] (v4) [label=180:$v_4$] {\small{$5$}};

\foreach \source/ \dest /\weight in {u/v1/a, u/x1/b, u/v2/c, u/x2/d, u/v3/e, u/x3/f, u/v4/g, u/x4/h, v3/z/q, v3/y/r}
\draw (\source) edge node[font=\small,pos=0.6] {$\weight$} (\dest);

\foreach \source/ \dest /\weight in {x4/v1/i, v1/x1/j, x1/v2/k, v2/x2/l, x2/v3/m, v3/x3/n, x3/v4/o, v4/x4/p, y/x3/s}
\draw (\source) edge node[font=\small,pos=0.5] {$\weight$} (\dest);

\draw (u) edge node[pos=0.2] {$u$} (v2);
\end{tikzpicture}
\label{fig:c7b}
}
\qquad
\subfloat[][]{
\begin{tikzpicture}[scale=1.25,auto]
\tikzstyle{whitenode}=[draw,circle,fill=white,minimum size=8pt,inner sep=0pt]
\tikzstyle{blacknode}=[draw,circle,fill=black,minimum size=6pt,inner sep=0pt]
\tikzstyle{tnode}=[draw,ellipse,fill=white,minimum size=8pt,inner sep=0pt]
\tikzstyle{texte} =[fill=white, text=black]
\foreach \pos/\name/\angle/\nom in {{(40:0cm)/u/5/}, {(90:2cm)/v1/90/v_1}, {(-90:2cm)/v3/-90/v_2}}
	\node[blacknode] (\name) [label=\angle: $\nom$] at \pos {};

\foreach \pos/\name/\angle/\nom in {{(45:2cm)/x1/45/x_1}, {(135:2cm)/x4/135/x_4}}
	\node[whitenode] (\name) [label=\angle: $\nom$] at \pos {};

\draw (v3)
 ++(-45:1.414cm) node[whitenode] (z) {};
 \draw (v3)
 ++(-135:1.414cm) node[whitenode] (y) [label=-135:$y$] {\small{$6$}};

\draw (-135:2cm) node[tnode] (x3) [label=-135:$x_3$] {\small{$6$}};
\draw (-45:2cm) node[tnode] (x2) [label=-45:$x_2$] {};
\draw (0:2cm) node[tnode] (v2) [label=0:$v_3$] {\small{$5$}};
\draw (180:2cm) node[tnode] (v4) [label=180:$v_4$] {\small{$5$}};

\foreach \source/ \dest /\weight in {u/v1/a, u/x1/b, u/v2/c, u/x2/d, u/v3/e, u/x3/f, u/v4/g, u/x4/h, v3/z/q, v3/y/r}
\draw (\source) edge node[font=\small,pos=0.6] {$\weight$} (\dest);

\foreach \source/ \dest /\weight in {x4/v1/i, v1/x1/j, x1/v2/k, v2/x2/l, x2/v3/m, v3/x3/n, x3/v4/o, v4/x4/p, y/x3/s}
\draw (\source) edge node[font=\small,pos=0.5] {$\weight$} (\dest);

\draw (u) edge node[pos=0.2] {$u$} (v2);
\end{tikzpicture}
\label{fig:c7c}
}
\caption{Notations of Claim~\ref{claim:C7}}
\label{fig:c7}
\end{figure}
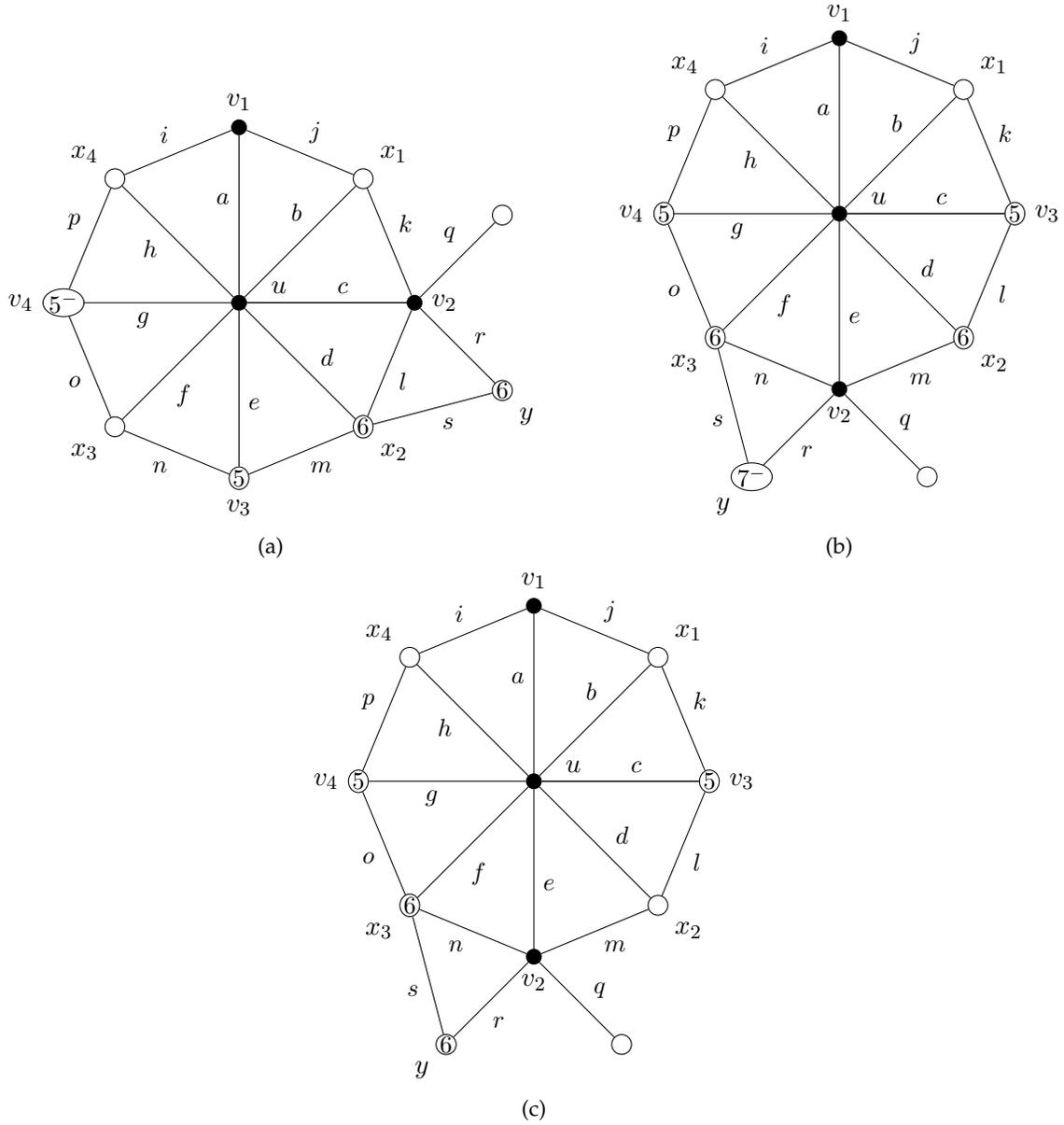
Note that since $G$ is a simple graph, the edges $s$ and $m$ cannot be incident. Since $|\hat{s}|+|\hat{m}|>|\hat{r}|$, there exists $\alpha \in (\hat{s} \cap \hat{m})\cup((\hat{s} \cup \hat{m})\setminus \hat{r})$. We color $m$ and $s$ in $\alpha$ if possible, in an arbitrary color otherwise. We can delete successively $r$, $n$ and $o$. We color successively $q$, $l$ and $k$. Note that $p$ and $j$ cannot be incident. Since $|\hat{p}|+|\hat{j}|>|\hat{i}|$, there exists $\beta \in (\hat{p} \cap \hat{j})\cup((\hat{p} \cup \hat{j})\setminus \hat{i})$. Thus we color $p$ and $j$ in $\beta$ if possible, in an arbitrary color otherwise. We delete successively $i$, $a$, $e$, $f$, $g$, $c$, $h$, $b$ and $d$.
\end{itemize}
\end{itemize}
\end{proof}

\begin{claim}\label{claim:C8}
$G$ cannot contain \textbf{($C_8$)}.
\end{claim}
\begin{proof}
Not that since $G$ is simple and $d(y) \neq d(v)$, all the vertices named here are distinct. We name the edges according to Figure~\ref{fig:c8}. By minimality, we color $G\setminus\{a,\ldots,f\}$. Without loss of generality, we consider the worst case, i.e. $|\hat{a}|=|\hat{d}|=|\hat{f}|=2$, $|\hat{c}|=|\hat{e}|=3$ and $|\hat{b}|=4$.\\
\begin{figure}[!h]
\centering
\begin{tikzpicture}[scale=2,auto,rotate=90]
\tikzstyle{whitenode}=[draw,circle,fill=white,minimum size=8pt,inner sep=0pt]
\tikzstyle{blacknode}=[draw,circle,fill=black,minimum size=6pt,inner sep=0pt]
\tikzstyle{tnode}=[draw,ellipse,fill=white,minimum size=8pt,inner sep=0pt]
\tikzstyle{texte} =[fill=white, text=black]

\draw (0,0) node[tnode] (u) [label=-90:$u$] {\small{$7$}}
 ++(45:1cm) node[tnode] (v) [label=left:$v$] {\small{$5$}};
\draw (u)
 ++(0:1.414cm) node[tnode] (w) [label=90:$w$] {\small{$6$}};
\draw (u)
 ++(-45:1cm) node[tnode] (x) [label=left:$x$] {\small{$5$}};

\draw (x)
 ++(-90:1cm) node[tnode] (y) [label=right:$y$] {\small{$6$}};

\foreach \source/ \dest /\weight in {x/y/a, w/x/b, v/w/c, u/v/d, x/u/e, u/w/f}
\draw (\source) edge node[font=\small,pos=0.5] {$\weight$} (\dest);
\end{tikzpicture}
\caption{Notations of Claim~\ref{claim:C8}}
\label{fig:c8}
\end{figure}
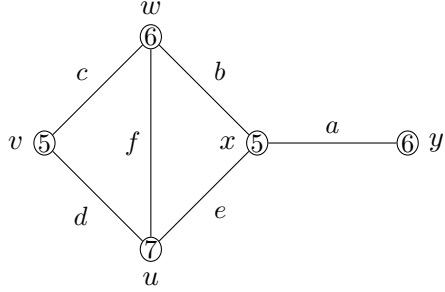

We consider two cases depending on whether $\hat{f}=\hat{d}$.
\begin{itemize}
\item $\hat{f} \neq \hat{d}$. We color $f$ in a color that does not belong to $\hat{d}$. We apply Lemma~\ref{lem:2322} on $(a,b,c,d,e)$.
\item $\hat{f}=\hat{d}$. We color $e$ and $c$ in a color that does not belong to $\hat{f}$. We color successively $a$, $b$, $f$ and $d$.
\end{itemize}
\end{proof}

\begin{claim}\label{claim:C9}
$G$ cannot contain \textbf{($C_9$)}.
\end{claim}
\begin{proof}
Note that by Claim~\ref{claim:C1}, $\{v_1,v_2,v_3\}$ forms a stable set. We consider two cases depending on whether there are two weak neighbors $w_1$ and $w_2$ of $u$ with $d(w_1)=d(w_2)=4$ and a vertex $x$, such that $(w_1,x,u)$ and $(w_2,x,u)$ are faces. 
\begin{itemize}
\item \emph{There are two weak neighbors $w_1$ and $w_2$ of $u$ with $d(w_1)=d(w_2)=4$ and a vertex $x$, such that $(w_1,x,u)$ and $(w_2,x,u)$ are faces}.\\We assume w.l.o.g. that the neighborhood of $u$ is, clockwise, $y_1$, $v_1$, $x$, $v_2$, $y_2$, $v_3$ and $z$. We are in one of the following three cases: either $d(z)\leq 7$, or $d(y_2)\leq 7$, or $d(z)=d(y_2)=8$.
\begin{itemize}
\item $d(z)\leq 7$.\\We name the edges according to Figure~\ref{fig:c9b}. By minimality, we color $G\setminus\{a,\ldots,o\}$. Without loss of generality, we consider the worst case, i.e. $|\hat{i}|=|\hat{l}|=|\hat{n}|=|\hat{o}|=2$, $|\hat{a}|=|\hat{h}|=3$, $|\hat{c}|=|\hat{e}|=|\hat{g}|=|\hat{j}|=|\hat{k}|=|\hat{m}|=4$, $|\hat{f}|=7$ and $|\hat{b}|=|\hat{d}|=9$. We first prove the following.
\begin{claimb}\label{clb:rec9a}
We can color $a,c,e,f,g,h,i,j,k,l,m,n$ and $o$ in such a way that $\hat{b} \neq \hat{d}$ if $|\hat{b}|=|\hat{d}|=1$.
\end{claimb}
\begin{proofclaim}
Let $B=\hat{b}$ and $D=\hat{d}$. If $\hat{b}=\hat{d}$, then we consider $\alpha \in \hat{i}$, and color $i$ in $\alpha$. If $\hat{b}\neq\hat{d}$, then we consider $\alpha \in \hat{d} \setminus \hat{b}$, and color $i$ arbitrarily. 
We remove color $\alpha$ from $\hat{k}$, $\hat{l}$ and $\hat{m}$.
We color $l$. We consider two cases depending on whether $\hat{m}=\hat{n}$.
\begin{itemize}
\item $\hat{n}\neq \hat{m}$.\\Then we color $m$ in a color that does not belong to $\hat{n}$. We color $k$. Since $|\hat{h}|+|\hat{c}|>|\hat{a}|$, there exists $\beta \in (\hat{h} \cap \hat{c})\cup((\hat{h} \cup \hat{c})\setminus \hat{a})$. We color $h$ and $c$ in $\beta$ if possible. We color successively $j$, and $h$ or $c$ if not colored already. We apply Lemma~\ref{lem:2322} on $(a,g,o,n,e)$. We color $f$.
\item $\hat{n}=\hat{m}$.\\Since $|\hat{a}|+|\hat{o}|>|\hat{g}|$, there exists $\beta \in (\hat{a} \cap \hat{o})\cup ((\hat{a} \cup \hat{o}) \setminus \hat{g})$. If $\beta \in \hat{e} \setminus \hat{m}$ or $\beta \not\in \hat{a}$, we color $e$ and $o$ in $\beta$ if possible, in an arbitrary color otherwise ($\not\in \hat{m}$ in the case of $e$). We color $n, m$ and $k$. We delete $g$, and we apply Lemma~\ref{lem:evencycle} on $(h,j,c,a)$. If $\beta \not\in \hat{e} \setminus \hat{m}$ and $\beta \in \hat{a}$, we color $a$ and $o$ in $\beta$ if possible, in an arbitrary color otherwise. Note that $a$ is colored in $\beta$, and that $\beta$ does not belong to $\hat{e}$ or belongs to $\hat{m}$, in which case one of $\{m,n\}$ will be colored in $\beta$. We color successively $n, m, k, h, j, c$, and $e$. We color $g$, and $f$.
\end{itemize}
Assume $|\hat{b}|=|\hat{d}|=1$. Then the colors of the edges incident to $b$ are all different and belong to $B$.
We consider two cases depending on whether $B=D$.
\begin{itemize}
\item \emph{$B=D$}. Since $i$ is colored in $\alpha$, no edge in $\{a,c,e,g\}$ is colored in $\alpha$, and $\alpha \in D$. By construction, none of $\{k,l,m\}$ is colored in $\alpha$. Thus $\alpha \in \hat{d}$ and $\alpha \not\in \hat{b}$, so $\hat{b}\neq\hat{d}$.
\item \emph{$B \neq D$}. Since $\alpha \not\in B$, no edge in $\{a,c,e,g\}$ is colored in $\alpha$. By construction, none of $\{k,l,m\}$ is colored in $\alpha$. Thus $\alpha \in \hat{d}$ and $\alpha \not\in \hat{b}$, so $\hat{b}\neq\hat{d}$.
\end{itemize}  
\end{proofclaim}

By (\ref{clb:rec9a}), we color $a,c,e,f,g,h,i,j,k,l,m,n$ and $o$ in such a way that $\hat{b} \neq \hat{d}$ if $|\hat{b}|=|\hat{d}|=1$. We color $b$ and $d$.
 
\item $d(y_2)\leq 7$.\\We name the edges according to Figure~\ref{fig:c9a}. By minimality, we color $G\setminus\{a,\ldots,m\}$. Without loss of generality, we consider the worst case, i.e. $|\hat{g}|=|\hat{i}|=|\hat{l}|=2$, $|\hat{a}|=|\hat{h}|=3$, $|\hat{c}|=|\hat{e}|=|\hat{j}|=|\hat{k}|=|\hat{m}|=4$, $|\hat{f}|=5$ and $|\hat{b}|=|\hat{d}|=9$. We first prove the following.
\begin{claimb}\label{clb:rec9b}
We can color $a,c,e,f,g,h,i,j,k,l$ and $m$ in such a way that, afterwards, $\hat{b} \neq \hat{d}$ if $|\hat{b}|=|\hat{d}|=1$.
\end{claimb}
\begin{proofclaim}
If $\hat{b}=\hat{d}$, then we consider $\alpha \in \hat{l}$, and color $l$ in $\alpha$. If $\hat{b}\neq\hat{d}$, then we consider $\alpha \in \hat{b} \setminus \hat{d}$, and color $l$ arbitrarily. We remove color $\alpha$ from $\hat{h}$, $\hat{i}$ and $\hat{j}$. We color successively $i, h, j, a, g, c, k, e, m$ and $f$. By the same analysis as in the previous case, $\hat{b} \neq \hat{d}$ if $|\hat{b}|=|\hat{d}|=1$.
\end{proofclaim}

By (\ref{clb:rec9b}), we color $a,c,e,f,g,h,i,j,k,l$ and $m$ in such a way that $\hat{b} \neq \hat{d}$ if $|\hat{b}|=|\hat{d}|=1$. We color $b$ and $d$.

\item $d(z)=d(y_2)=8$.\\Then either $v_3$ is a weak neighbor of $u$ of degree $4$, or $v_3$ is a weak neighbor of $u$ of degree $5$ adjacent to a vertex of degree $6$. We will deal with the two cases at once. We consider that $v_3$ is of degree $5$ in both cases, by adding a neighbor of degree $6$ to $v_3$ if it is of degree $4$: a proper coloring of this graph will yield a proper coloring of the initial graph. We name the edges according to Figure~\ref{fig:c9a}.
 
By minimality, we color $G\setminus\{a,\ldots,q\}$. Without loss of generality, we consider the worst case, i.e. $|\hat{i}|=|\hat{l}|=|\hat{p}|=2$, $|\hat{a}|=|\hat{g}|=|\hat{h}|=|\hat{o}|=3$, $|\hat{c}|=|\hat{e}|=|\hat{j}|=|\hat{k}|=|\hat{m}|=|\hat{n}|=|\hat{q}|=4$ and $|\hat{b}|=|\hat{d}|=|\hat{f}|=9$. We first prove the following.
\begin{claimb}\label{clb:rec9c}
We can color $a,c,e,f,g,h,i,j,k,l,m,n,o,p$ and $q$ in such a way that, afterwards, $\hat{b} \neq \hat{d}$ if $|\hat{b}|=|\hat{d}|=1$.
\end{claimb}
\begin{proofclaim}
If $\hat{b}=\hat{d}$, then we consider $\alpha \in \hat{l}$, and color $i$ in $\alpha$. If $\hat{b}\neq\hat{d}$, then we consider $\alpha \in \hat{d} \setminus \hat{b}$, and color $i$ arbitrarily. We remove color $\alpha$ from $\hat{k}$, $\hat{l}$ and $\hat{m}$.

We color $l$. Since $|\hat{g}|+|\hat{p}|>|\hat{o}|$, there exists $\beta \in (\hat{g} \cap \hat{p})\cup ((\hat{g} \cup \hat{p})\setminus \hat{o})$. We color $g$ and $p$ in $\beta$ if possible, in an arbitrary color otherwise. We color $m$ so that $\hat{e}\neq\hat{a}$ if $|\hat{a}|=|\hat{e}|=2$, which is possible as $|\hat{m}|\geq 2$. We color $k$, and we apply Lemma~\ref{lem:2322} on $(e,a,h,j,c)$. We color $n, o, q$ and $f$.

By the same analysis as in the two previous cases, we have $\hat{b}\neq\hat{d}$ if $|\hat{b}|=|\hat{d}|=1$.
\end{proofclaim}

By (\ref{clb:rec9c}), we color $a,c,e,f,g,h,i,j,k,l,m,n,o,p$ and $q$ in such a way that $\hat{b} \neq \hat{d}$ if $|\hat{b}|=|\hat{d}|=1$. We color $b$ and $d$.
\end{itemize}

\item \emph{There are no two weak neighbors $w_1$ and $w_2$ of $u$ with $d(w_1)=d(w_2)=4$ for which there exists a vertex $x$ such that $(w_1,x,u)$ and $(w_2,x,u)$ are faces}.\\Then $v_3$ must be a vertex of degree $5$. W.l.o.g., the neighborhood of $u$ is, clockwise, $y_1$, $v_1$, $y_2$, $v_3$, $y_3$, $v_2$, $y_4$. We consider two cases depending on whether $d(y_2)=d(y_3)=8$.
\begin{itemize}
\item \emph{$d(y_2)\leq 7$ or $d(y_3) \leq 7$}.\\Consider w.l.o.g. that $d(y_2) \leq 7$. We name the edges according to Figure~\ref{fig:c9e}. By minimality, we color $G\setminus\{a,\ldots,o\}$. Without loss of generality, we consider the worst case, i.e. $|\hat{i}|=|\hat{l}|=|\hat{n}|=2$, $|\hat{a}|=|\hat{g}|=|\hat{h}|=|\hat{k}|=|\hat{o}|=3$, $|\hat{e}|=|\hat{m}|=4$, $|\hat{c}|=|\hat{j}|=5$, $|\hat{d}|=7$ and $|\hat{b}|=|\hat{f}|=9$. \\

Note that the edges $k$ and $h$ are not incident. Since $|\hat{k}|+|\hat{h}|>|\hat{j}|$, there exists $\alpha \in (\hat{k} \cap \hat{h})\cup((\hat{k} \cup \hat{h})\setminus \hat{j})$. We color $k$ and $h$ in $\alpha$ if possible, in an arbitrary color otherwise. We can delete successively $j$, $b$, $i$, $f$, $d$ and $c$. We color $a$, $l$ and $n$. We apply Lemma~\ref{lem:evencycle} on $(e,m,o,g)$.

\item \emph{$d(y_2)=d(y_3)=8$}.\\Then $v_3$ must be a weak neighbor of degree $5$ whose two other neighbors are of degree $6$ and $7$, respectively. We name the edges according to Figure~\ref{fig:c9f}. By minimality, we color $G\setminus\{a,\ldots,q\}$. Without loss of generality, we consider the worst case, i.e. $|\hat{i}|=|\hat{n}|=2$, $|\hat{a}|=|\hat{g}|=|\hat{h}|=|\hat{o}|=|\hat{q}|=3$, $|\hat{c}|=|\hat{e}|=|\hat{j}|=|\hat{k}|=|\hat{l}|=|\hat{m}|=|\hat{p}|=4$, and $|\hat{b}|=|\hat{d}|=|\hat{f}|=9$. We first prove the following.
\begin{figure}
\centering
\subfloat[][]{
\begin{tikzpicture}[scale=1.2,auto]
\tikzstyle{whitenode}=[draw,circle,fill=white,minimum size=8pt,inner sep=0pt]
\tikzstyle{blacknode}=[draw,circle,fill=black,minimum size=6pt,inner sep=0pt]
\tikzstyle{tnode}=[draw,ellipse,fill=white,minimum size=8pt,inner sep=0pt]
\tikzstyle{texte} =[fill=white, text=black]
\foreach \pos/\name/\angle/\nom in {{(40:0cm)/u/5/}, {(45:2cm)/v2/90/v_2}, {(135:2cm)/v1/90/v_1}}
	\node[blacknode] (\name) [label=\angle: $\nom$] at \pos {};

\foreach \pos/\name/\angle/\nom in {{(180:2cm)/y1/180/y_1}, {(135:3.2cm)/w1/0/}, {(90:2cm)/x/90/x}, {(45:3.2cm)/w2/180/}, {(0:2cm)/y2/0/y_2}}
	\node[whitenode] (\name) [label=\angle: $\nom$] at \pos {};


\draw (-90:2cm) node[tnode] (z) [label=-90:$z$] {\small{$7^-$}};
\draw (-45:2cm) node[tnode] (v3) [label=-90:$v_3$] {\small{$5^-$}};

\foreach \source/ \dest /\weight in {u/y1/a, u/v1/b, u/x/c, u/v2/d, u/y2/e, u/v3/f, u/z/g, v1/w1/i, v2/w2/l}
\draw (\source) edge node[font=\small,pos=0.6] {$\weight$} (\dest);

\foreach \source/ \dest /\weight in {y1/v1/h, v1/x/j, x/v2/k, v2/y2/m, y2/v3/n, v3/z/o}
\draw (\source) edge node[font=\small,pos=0.5] {$\weight$} (\dest);

\draw (u) edge node[pos=0.2] {$u$} (y2);
\end{tikzpicture}
\label{fig:c9b}
}
\qquad
\subfloat[][]{
\begin{tikzpicture}[scale=1.2,auto]
\tikzstyle{whitenode}=[draw,circle,fill=white,minimum size=8pt,inner sep=0pt]
\tikzstyle{blacknode}=[draw,circle,fill=black,minimum size=6pt,inner sep=0pt]
\tikzstyle{tnode}=[draw,ellipse,fill=white,minimum size=8pt,inner sep=0pt]
\tikzstyle{texte} =[fill=white, text=black]
\foreach \pos/\name/\angle/\nom in {{(40:0cm)/u/5/}, {(45:2cm)/v2/90/v_2}, {(135:2cm)/v1/90/v_1}}
	\node[blacknode] (\name) [label=\angle: $\nom$] at \pos {};

\foreach \pos/\name/\angle/\nom in {{(180:2cm)/y1/180/y_1}, {(135:3.2cm)/w1/0/}, {(90:2cm)/x/90/x}, {(45:3.2cm)/w2/180/}, {(-90:2cm)/z/-90/z}}
	\node[whitenode] (\name) [label=\angle: $\nom$] at \pos {};


\draw (0:2cm) node[tnode] (y2) [label=0:$y_2$] {\small{$7^-$}};
\draw (-45:2cm) node[tnode] (v3) [label=-90:$v_3$] {\small{$5^-$}};

\foreach \source/ \dest /\weight in {u/y1/a, u/v1/b, u/x/c, u/v2/d, u/y2/e, u/v3/f, u/z/g, v1/w1/i, v2/w2/l}
\draw (\source) edge node[font=\small,pos=0.6] {$\weight$} (\dest);

\foreach \source/ \dest /\weight in {y1/v1/h, v1/x/j, x/v2/k, v2/y2/m,v3/y2/, v3/z/}
\draw (\source) edge node[font=\small,pos=0.5] {$\weight$} (\dest);

\draw (u) edge node[pos=0.2] {$u$} (y2);
\end{tikzpicture}
\label{fig:c9a}
}
\qquad
\subfloat[][]{
\begin{tikzpicture}[scale=1.2,auto]
\tikzstyle{whitenode}=[draw,circle,fill=white,minimum size=8pt,inner sep=0pt]
\tikzstyle{blacknode}=[draw,circle,fill=black,minimum size=6pt,inner sep=0pt]
\tikzstyle{tnode}=[draw,ellipse,fill=white,minimum size=8pt,inner sep=0pt]
\tikzstyle{texte} =[fill=white, text=black]
\foreach \pos/\name/\angle/\nom in {{(40:0cm)/u/5/}, {(45:2cm)/v2/90/v_2}, {(135:2cm)/v1/90/v_1}, {(-45:2cm)/v3/-45/v_3}}
	\node[blacknode] (\name) [label=\angle: $\nom$] at \pos {};

\foreach \pos/\name/\angle/\nom in {{(180:2cm)/y1/180/y_1}, {(135:3.2cm)/w1/0/}, {(90:2cm)/x/90/x}, {(45:3.2cm)/w2/180/}, {(0:2cm)/y2/0/y_2}, {(-90:2cm)/z/-90/z}}
	\node[whitenode] (\name) [label=\angle: $\nom$] at \pos {};

\draw (v3)
 ++(0:1.414cm) node[whitenode] (a) {};
 \draw (v3)
 ++(-90:1.414cm) node[whitenode] (b) {\small{$6$}};


\foreach \source/ \dest /\weight in {u/y1/a, u/v1/b, u/x/c, u/v2/d, u/y2/e, u/v3/f, u/z/g, v1/w1/i, v2/w2/l}
\draw (\source) edge node[font=\small,pos=0.6] {$\weight$} (\dest);

\foreach \source/ \dest /\weight/\test in {v3/a/p/swap, v3/b/q/}
\draw[\test] (\source) edge node[font=\small,pos=0.7] {$\weight$} (\dest);

\foreach \source/ \dest /\weight in {y1/v1/h, v1/x/j, x/v2/k, v2/y2/m, y2/v3/n, v3/z/o}
\draw (\source) edge node[font=\small,pos=0.5] {$\weight$} (\dest);

\draw (u) edge node[pos=0.2] {$u$} (y2);
\end{tikzpicture}
\label{fig:c9c}
}
\qquad
\subfloat[][]{
\begin{tikzpicture}[scale=1.2,auto]
\tikzstyle{whitenode}=[draw,circle,fill=white,minimum size=8pt,inner sep=0pt]
\tikzstyle{blacknode}=[draw,circle,fill=black,minimum size=6pt,inner sep=0pt]
\tikzstyle{tnode}=[draw,ellipse,fill=white,minimum size=8pt,inner sep=0pt]
\tikzstyle{texte} =[fill=white, text=black]
\foreach \pos/\name/\angle/\nom in {{(40:0cm)/u/5/}, {(135:2cm)/v1/90/v_1}, {(-45:2cm)/v3/right/v_2}}
	\node[blacknode] (\name) [label=\angle: $\nom$] at \pos {};

\foreach \pos/\name/\angle/\nom in {{(180:2cm)/y1/180/y_1}, {(135:3.2cm)/w1/0/}, {(-45:3.2cm)/w2/180/}, {(0:2cm)/y2/0/y_3}, {(-90:2cm)/z/-90/y_4}}
	\node[whitenode] (\name) [label=\angle: $\nom$] at \pos {};

\draw (45:2cm) node[tnode] (v2) [label=45:$v_3$] {\small{$5$}};
\draw (90:2cm) node[tnode] (x) [label=90:$y_2$] {\small{$7$}};

\foreach \source/ \dest /\weight in {u/y1/a, u/v1/b, u/x/c, u/v2/d, u/y2/e, u/v3/f, u/z/g, v1/w1/i, v3/w2/n}
\draw (\source) edge node[font=\small,pos=0.6] {$\weight$} (\dest);

\foreach \source/ \dest /\weight in {y1/v1/h, v1/x/j, x/v2/k, v2/y2/l, y2/v3/m, v3/z/o}
\draw (\source) edge node[font=\small,pos=0.5] {$\weight$} (\dest);

\draw (u) edge node[pos=0.2] {$u$} (y2);
\end{tikzpicture}
\label{fig:c9e}
}
\qquad
\subfloat[][]{
\begin{tikzpicture}[scale=1.2,auto]
\tikzstyle{whitenode}=[draw,circle,fill=white,minimum size=8pt,inner sep=0pt]
\tikzstyle{blacknode}=[draw,circle,fill=black,minimum size=6pt,inner sep=0pt]
\tikzstyle{tnode}=[draw,ellipse,fill=white,minimum size=8pt,inner sep=0pt]
\tikzstyle{texte} =[fill=white, text=black]
\foreach \pos/\name/\angle/\nom in {{(40:0cm)/u/5/}, {(45:2cm)/v2/45/v_3}, {(135:2cm)/v1/90/v_1}, {(-45:2cm)/v3/0/v_2}}
	\node[blacknode] (\name) [label=\angle: $\nom$] at \pos {};

\foreach \pos/\name/\angle/\nom in {{(180:2cm)/y1/180/y_1}, {(135:3.2cm)/w1/0/}, {(90:2cm)/x/90/y_2}, {(-45:3.2cm)/w2/180/}, {(0:2cm)/y2/0/y_3}, {(-90:2cm)/z/-90/y_4}}
	\node[whitenode] (\name) [label=\angle: $\nom$] at \pos {};

\draw (v2)
 ++(90:1.414cm) node[tnode] (a) {\small{$6$}};
 \draw (v2)
 ++(0:1.414cm) node[tnode] (b) {\small{$7$}};


\foreach \source/ \dest /\weight in {u/y1/a, u/v1/b, u/x/c, u/v2/d, u/y2/e, u/v3/f, u/z/g, v1/w1/i, v3/w2/n}
\draw (\source) edge node[font=\small,pos=0.6] {$\weight$} (\dest);

\foreach \source/ \dest /\weight/\test in {v2/a/p/swap, v2/b/q/}
\draw[\test] (\source) edge node[font=\small,pos=0.7] {$\weight$} (\dest);

\foreach \source/ \dest /\weight in {y1/v1/h, v1/x/j, x/v2/k, v2/y2/l, y2/v3/m, v3/z/o}
\draw (\source) edge node[font=\small,pos=0.5] {$\weight$} (\dest);

\draw (u) edge node[pos=0.2] {$u$} (y2);
\end{tikzpicture}
\label{fig:c9f}
}
\caption{Notations of Claim~\ref{claim:C9}}
\label{fig:c9}
\end{figure}
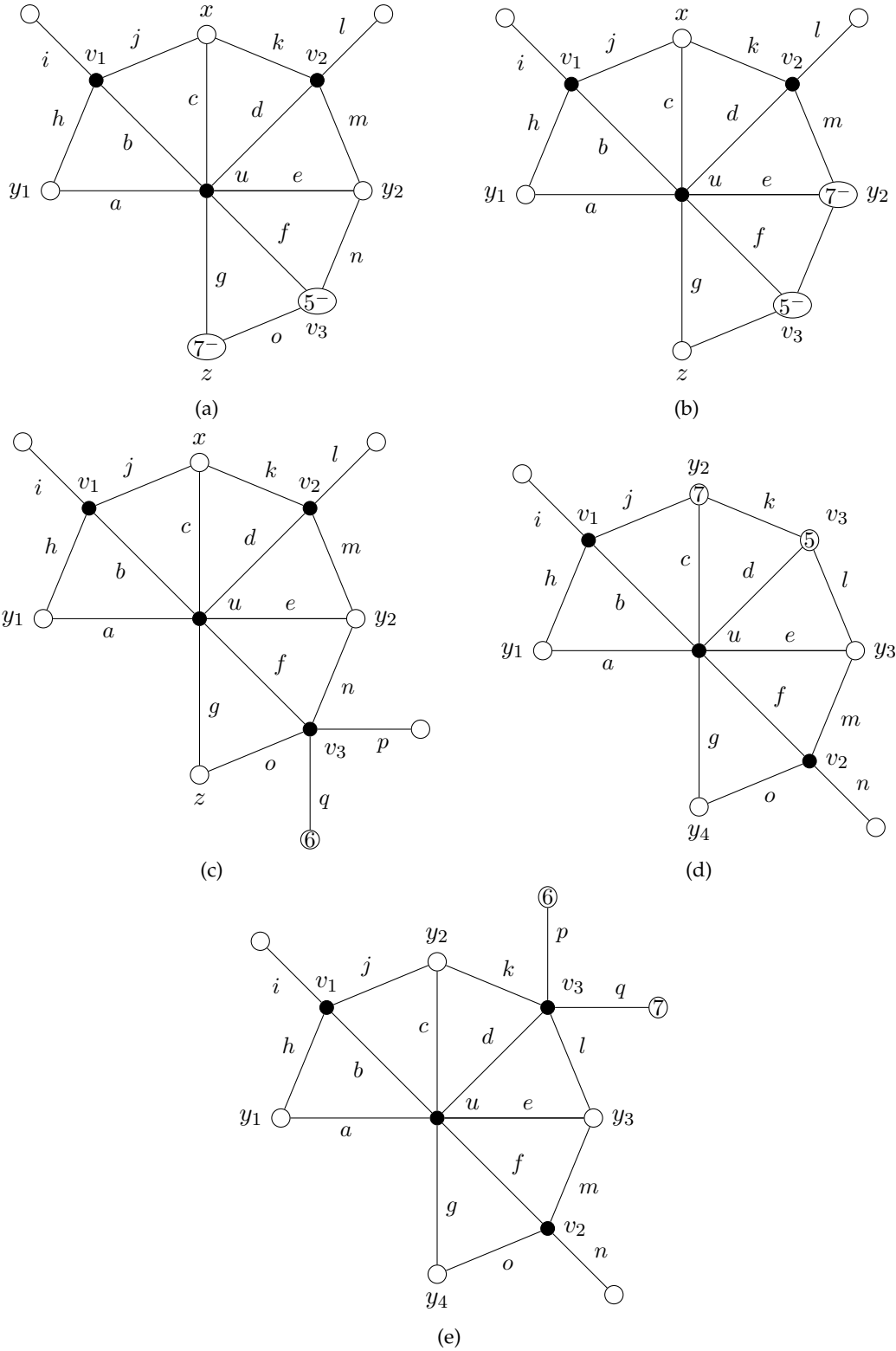
\begin{claimb}\label{clb:rec9d}
We can color $a,c,e,f,g,h,i,j,k,l,m,n,o,p$ and $q$ in such a way that, afterwards, $\hat{b} \neq \hat{f}$ if $|\hat{b}|=|\hat{f}|=1$.
\end{claimb}
\begin{proofclaim}
If $\hat{b}=\hat{f}$, then we consider $\alpha \in \hat{n}$, and color $n$ in $\alpha$. If $\hat{b}\neq\hat{f}$, then we consider $\alpha \in \hat{b} \setminus \hat{d}$, and color $n$ arbitrarily. We remove color $\alpha$ from $\hat{h}$, $\hat{i}$ and $\hat{j}$.
We color successively $i, h, j, k, c, a, g, e, o, m, l, q, p$ and $d$.
By the same analysis as in the previous cases, we have $\hat{b}\neq\hat{f}$ if $|\hat{b}|=|\hat{f}|=1$.
\end{proofclaim}

By (\ref{clb:rec9d}), we color $a,c,e,f,g,h,i,j,k,l,m,n,o,p$ and $q$ in such a way that $\hat{b} \neq \hat{f}$ if $|\hat{b}|=|\hat{f}|=1$. We color $b$ and $f$.
\end{itemize}
\end{itemize}
\end{proof}

\begin{claim}\label{claim:C10}
$G$ cannot contain \textbf{($C_{10}$)}.
\end{claim}
\begin{proof}
We consider two cases depending on whether $v_2$ and $u$ have a common neighbor of degree $6$. 
\begin{itemize}
\item \emph{Vertices $v_2$ and $u$ have a common neighbor $y$ of degree $6$}.\\By definition of an $S_3$-neighbor, vertex $v_2$ has two other neighbors of degree $7$ and $6$, respectively. We name the edges according to Figure~\ref{fig:c10a}.
Since the graph is simple, there is no $1 \leq i \leq 3$ such that the edges $e$ and $c_i$ are incident.

By minimality, we color $G\setminus\{a,b_1,b_2,c_1,c_2,c_3,d,e\}$. Without loss of generality, we consider the worst case, i.e. $|\hat{b_1}|=|\hat{c_1}|=|\hat{e}|=2$, $|\hat{b_2}|=|\hat{c_2}|=3$, $|\hat{c_3}|=4$, $|\hat{d}|=5$, and $|\hat{a}|=6$.

Since $|\hat{e}|+|\hat{c_3}|>|\hat{d}|$, there exists $\alpha \in (\hat{e} \cap \hat{c_3})\cup ((\hat{e} \cup \hat{c_3}) \setminus \hat{d})$. If $\alpha \in \hat{c_3}$, let $i$ be the minimum integer such that $\alpha \in \hat{c_i}$. If $\alpha \not\in \hat{c_3}$, then $\alpha \in \hat{e} \setminus (\hat{d}\cup \hat{c_3})$, let $i$ be $1$. We color $e$ and $c_i$ in $\alpha$ if possible, in an arbitrary color otherwise (by choice of $i$, if $c_i$ is not colored in $\alpha$ then $i=1$). We delete $d$. If $i \neq 3$, edge $c_3$ is not colored and we delete it. Then, if $i \neq 2$, edge $c_2$ is not colored, and either $i=3$ and $c_3$ is colored in $\alpha$ (which was not an available color for $c_2$ by choice of $i$), or $i=1$ and $c_3$ has been deleted; In both cases, we can delete $c_2$. Then, if $i \neq 1$, edge $c_1$ is not colored, and the edges $c_2$ and $c_3$ are deleted or colored in $\alpha$ (which was not an available color for $c_1$ by choice of $i$), so we can delete $c_1$. We delete successively $a$, $b_2$, $b_1$.

\item \emph{Vertices $v_2$ and $u$ have no common neighbor of degree $6$}.\\Then, by definition of an $S_3$-vertex, the neighborhood of $v_2$ is, clockwise, $(u,y_1,z_1,z_2,y_2)$, with $d(y_1)=d(y_2)=7$ and $d(z_1)=d(z_2)=6$. We name the edges according to Figure~\ref{fig:c10b}. By minimality, we color $G\setminus\{a,\ldots,k\}$. Without loss of generality, we consider the worst case, i.e. $|\hat{f}|=|\hat{g}|=|\hat{h}|=|\hat{j}|=2$, $|\hat{i}|=|\hat{k}|=3$, $|\hat{b}|=|\hat{e}|=5$, and $|\hat{a}|=|\hat{c}|=|\hat{d}|=6$. We first prove the following.
\begin{figure}[!h]
\centering
\subfloat[][]{
\begin{tikzpicture}[scale=2,auto]
\tikzstyle{whitenode}=[draw,circle,fill=white,minimum size=8pt,inner sep=0pt]
\tikzstyle{blacknode}=[draw,circle,fill=black,minimum size=6pt,inner sep=0pt]
\tikzstyle{invisible}=[draw=white,circle,fill=white,minimum size=6pt,inner sep=0pt]
\tikzstyle{tnode}=[draw,ellipse,fill=white,minimum size=8pt,inner sep=0pt]
\tikzstyle{texte} =[fill=white, text=black]

\draw (0,0) node[blacknode] (u) [label=right:$u$] {}
 ++(90:1cm) node[blacknode] (w) [label=90:$v_2$] {};
\draw (u)
 ++(45:1cm) node[tnode] (y) [label=right:$y$] {\small{$6$}};
\draw (u)
-- ++(180:1cm) node[whitenode] (y4) {};
\draw (u)
-- ++(-90:1cm) node[tnode] (v) [label=-90:$v_1$] {\small{$4$}};

\draw (u)
 ++(-135+90:1cm) node[tnode] (x) [label=-90:$v_3$] {\small{$5^-$}};
\draw (u)
-- ++(-90-30:1cm) node[whitenode] (y7) {};
\draw (u)
-- ++(-90-60:1cm) node[whitenode] (y8) {};

\draw (w)
 ++(-67.5-90:0.8cm) node[whitenode] (w4) {};
\draw (w)
 ++(-37.5+90:0.9cm) node[whitenode] (w2) {\small{$6$}};
\draw (w)
 ++(37.5+90:0.9cm) node[tnode] (w3) {\small{$7$}};

\foreach \source/ \dest /\weight/\test in {u/x/b_1/, u/v/b_2/, w/w4/c_1/swap, w/w3/c_2/swap, w/w2/c_3/swap}
\draw (\source) edge node[font=\small,pos=0.6,\test] {$\weight$} (\dest);

\foreach \source/ \dest /\weight in {u/w/a, y/u/e, w/y/d}
\draw (\source) edge node[font=\small,pos=0.5] {$\weight$} (\dest);

\end{tikzpicture}
\label{fig:c10a}
}
\qquad
\subfloat[][]{
\begin{tikzpicture}[scale=2,auto]
\tikzstyle{whitenode}=[draw,circle,fill=white,minimum size=8pt,inner sep=0pt]
\tikzstyle{blacknode}=[draw,circle,fill=black,minimum size=6pt,inner sep=0pt]
\tikzstyle{invisible}=[draw=white,circle,fill=white,minimum size=6pt,inner sep=0pt]
\tikzstyle{tnode}=[draw,ellipse,fill=white,minimum size=8pt,inner sep=0pt]
\tikzstyle{texte} =[fill=white, text=black]

\draw (0,0) node[blacknode] (u) [label=right:$u$] {}
 ++(90:1cm) node[blacknode] (w) [label=90:$v_2$] {};
\draw (u)
-- ++(45:1cm) node[tnode] (y3) [label=right:$y_2$] {\small{$7$}};
\draw (u)
-- ++(135:1cm) node[whitenode] (y4) [label=left:$y_1$] {\small{$7$}};
\draw (u)
-- ++(-90:1cm) node[tnode] (y5) [label=-90:$v_1$] {\small{$4$}};

\draw (u)
-- ++(-135:1cm) node[tnode] (y6) [label=-90:$v_3$] {\small{$5^-$}};
\draw (u)
-- ++(-30:1cm) node[whitenode] (y7) {};
\draw (u)
-- ++(-60:1cm) node[whitenode] (y8) {};

\draw (w)
 ++(-37.5+90:0.9cm) node[whitenode] (w2) [label=right:$z_2$] {\small{$6$}};
\draw (w)
 ++(37.5+90:0.9cm) node[tnode] (w3) [label=left:$z_1$] {\small{$6$}};
\draw (y4) edge node  {} (w3);
\draw (w2) edge node  {} (w3);
\draw (w2) edge node  {} (y3);

\foreach \source/ \dest /\weight in {w/y4/b, w/w3/c, w/w2/d, w/y3/e, u/y5/k}
\draw (\source) edge node[font=\small,pos=0.6] {$\weight$} (\dest);

\foreach \source/ \dest /\weight in {w/u/a, y3/u/f, u/y4/g, y4/w3/h, w3/w2/i, w2/y3/j}
\draw (\source) edge node[font=\small,pos=0.5] {$\weight$} (\dest);
\end{tikzpicture}
\label{fig:c10b}
}
\caption{Notations of Claim~\ref{claim:C10}}
\label{fig:c10}
\end{figure}
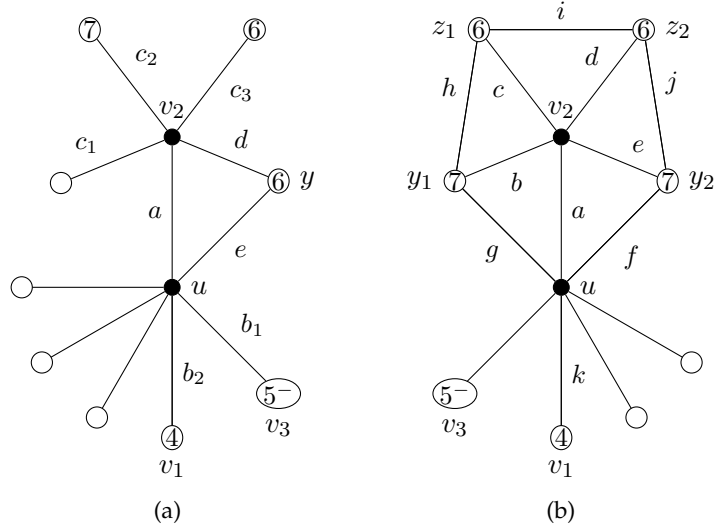
\begin{claimb}\label{clb:rec10b}
We can color $f,g,h,i,j$ and $k$ in such a way that, afterwards, $\hat{c} \neq \hat{d}$ if $|\hat{c}|=|\hat{d}|=4$.
\end{claimb}
\begin{proofclaim}
We consider two cases depending on whether $\hat{c}=\hat{d}$.
\begin{itemize}
\item $\hat{c}=\hat{d}$. We apply Lemma~\ref{lem:evencycle} on $(f,g,h,j)$ by considering that $f$ and $j$ are incident so they receive different colors. We color $i$ and $k$. The new constraints of $c$ are $i$ and $h$, and the new constraints of $d$ are $i$ and $j$. Since $h$ and $j$ receive distinct colors, we have $|\hat{c}|\geq 5$ or $ |\hat{d}|\geq 5$ or $\hat{c}\neq \hat{d}$.
\item $\hat{c} \neq \hat{d}$. Let $\alpha \in \hat{c}, \not\in \hat{d}$. We color $h$ in a color other than $\alpha$. We color $g, f, j, i$ and $k$ successively. Thus, either $|\hat{d}|\geq 5$ or $\alpha \in \hat{c}$ so $\hat{c} \neq \hat{d}$.
\end{itemize}
\end{proofclaim}

By (\ref{clb:rec10b}), we color $f,g,h,i,j$ and $k$ in such a way that $\hat{c} \neq \hat{d}$ if $|\hat{c}|=|\hat{d}|=4$. We color $a$, $b$ and $e$. Either $|\hat{c}| \geq 2$ (resp. $|\hat{d}| \geq 2$), and we color $d$ and $c$ (resp. $c$ and $d$). Or $|\hat{c}|=|\hat{d}|=1$ and $\hat{c} \neq \hat{d}$, we color $d$ and $c$ independently.
\end{itemize}
\end{proof}

\begin{claim}\label{claim:C11}
$G$ cannot contain \textbf{($C_{11}$)}.
\end{claim}
\begin{proof}
We name the edges according to Figure~\ref{fig:c11}. By minimality, we color $G\setminus\{a,\ldots,e\}$. Without loss of generality, we consider the worst case, i.e. $|\hat{d}|=|\hat{e}|=2$, $|\hat{a}|=|\hat{c}|=3$ and $|\hat{b}|=4$. We consider two cases depending on whether $\hat{e} \subset \hat{b}$.
\begin{figure}[!h]
\centering
\begin{tikzpicture}[scale=2,auto]
\tikzstyle{whitenode}=[draw,circle,fill=white,minimum size=8pt,inner sep=0pt]
\tikzstyle{blacknode}=[draw,circle,fill=black,minimum size=6pt,inner sep=0pt]
\tikzstyle{tnode}=[draw,ellipse,fill=white,minimum size=8pt,inner sep=0pt]
\tikzstyle{texte} =[fill=white, text=black]
 \draw (0,0) node[tnode] (u) [label=left:$u$] {\small{$5$}}
 ++(0:1cm) node[tnode] (w) [label=right:$w$] {\small{$6$}};
\draw (u)
 ++(60:1cm) node[tnode] (v) [label=90:$v$] {\small{$6$}};
\draw (u)
 ++(-60:1cm) node[tnode] (x) [label=-90:$x$] {\small{$6$}};

\foreach \source/ \dest /\weight in {u/v/a, v/w/d, w/x/e, x/u/c, u/w/b}
\draw (\source) edge node[font=\small,pos=0.5] {$\weight$} (\dest);
\end{tikzpicture}
\caption{Notations of Claim~\ref{claim:C11}}
\label{fig:c11}
\end{figure}
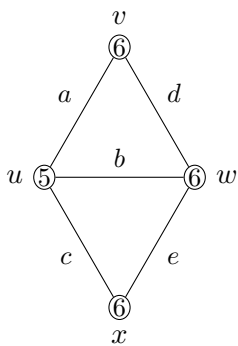
\begin{itemize}
\item $\hat{e} \not\subset \hat{b}$. Then we color $e$ in a color that does not belong to $\hat{b}$. We can delete successively $b$, $a$, $c$ and $d$.
\item $\hat{e} \subset \hat{b}$. Then, since $|\hat{c}|+|\hat{d}|>|\hat{b}|$ and $\hat{d} \subset \hat{b}$, there exists $\alpha \in (\hat{c} \cap \hat{d}) \cup ((\hat{c} \cup \hat{d})\setminus \hat{b})$. We color $c$ and $d$ in $\alpha$ if possible, in an arbitrary color otherwise. Note that since $\hat{e} \subset \hat{b}$, we have $|\hat{e}|\geq 1$ in both cases. We delete successively $b$, $e$ and $a$.
\end{itemize}
\end{proof}

Lemma~\ref{lem:config} holds by Claims~\ref{claim:C1} to \ref{claim:C11}.

\end{proof}

\section{Discharging rules}\label{sect:dis}

We design discharging rules $R_1$, $R_2$, $\ldots$, $R_{11}$ (see Figure~\ref{fig:dis}):\newline

For any face $f$ of degree at least $4$,
\begin{itemize}
\item Rule $R_1$ is when $d(f)=4$ and $f$ is incident to a vertex $v$ of degree $d(v)\leq 5$. Then $f$ gives $1$ to $v$.
\item Rule $R_2$ is when $d(f)\geq 5$ and $f$ is incident to a vertex $v$ of degree $d(v)\leq 5$. Then $f$ gives $2$ to $v$.
\end{itemize}

For any vertex $u$ of degree at least $7$,
\begin{itemize}
\item Rule $R_3$ is when $u$ has a weak neighbor $v$ of degree $3$. Then $u$ gives $1$ to $v$.
\item Rule $R_4$ is when $u$ has a semi-weak neighbor $v$ of degree $3$. Then $u$ gives $\frac{1}{2}$ to $v$.
\item Rule $R_5$ is when $u$ has a weak neighbor $v$ of degree $4$. Then $u$ gives $\frac{1}{2}$ to $v$.
\end{itemize}

For any vertex $u$ of degree $8$,
\begin{itemize}
\item Rule $R_6$ is when $u$ has an $E_2$-neighbor $v$. Then $u$ gives $\frac{1}{2}$ to $v$.
\item Rule $R_7$ is when $u$ has an $E_3$-neighbor $v$. Then $u$ gives $\frac{1}{3}$ to $v$.
\item Rule $R_8$ is when $u$ has an $E_4$-neighbor $v$. Then $u$ gives $\frac{1}{4}$ to $v$.
\end{itemize}

For any vertex $u$ of degree $7$,
\begin{itemize}
\item Rule $R_9$ is when $u$ has an $S_2$-neighbor $v$. Then $u$ gives $\frac{1}{2}$ to $v$.
\item Rule $R_{10}$ is when $u$ has an $S_3$-neighbor $v$. Then $u$ gives $\frac{1}{3}$ to $v$.
\item Rule $R_{11}$ is when $u$ has an $S_4$-neighbor $v$. Then $u$ gives $\frac{1}{4}$ to $v$.
\end{itemize}

\captionsetup[subfloat]{labelformat=empty}
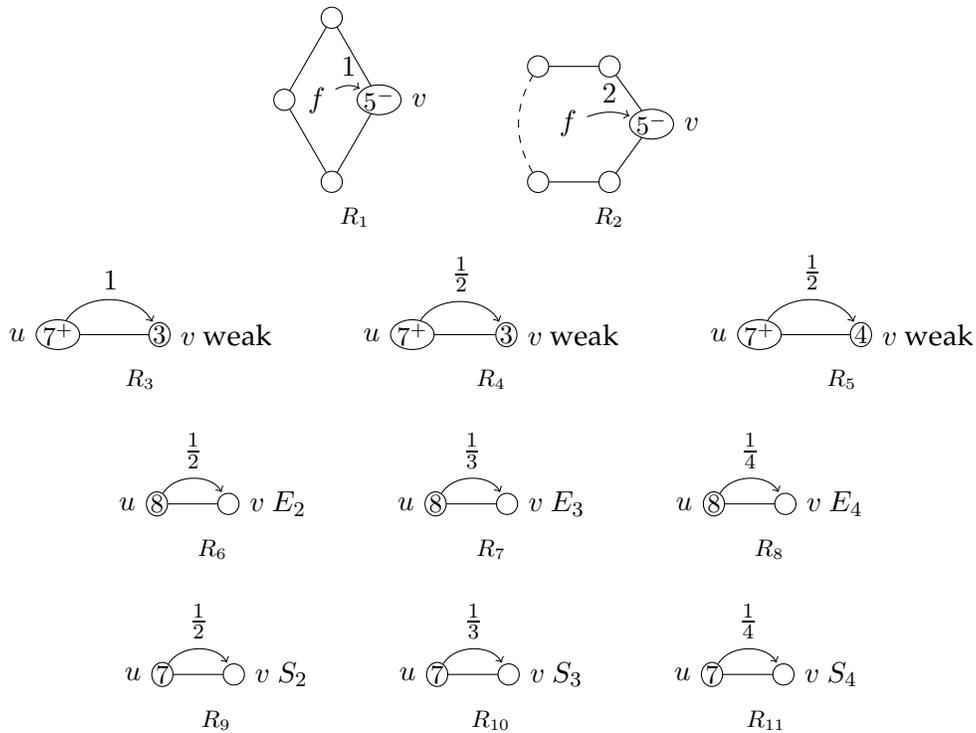
\begin{figure}[!h]
\centering
\subfloat[][$R_1$]{
\centering
\begin{tikzpicture}[scale=1.05]
\tikzstyle{whitenode}=[draw,circle,fill=white,minimum size=8pt,inner sep=0pt]
\tikzstyle{blacknode}=[draw,circle,fill=black,minimum size=6pt,inner sep=0pt]
\tikzstyle{tnode}=[draw,ellipse,fill=white,minimum size=8pt,inner sep=0pt]
\tikzstyle{texte} =[fill=white, text=black]
\draw (0,0) node[whitenode] (u) {}
-- ++(60:1.2cm) node[whitenode] (w) {}
-- ++(-60:1.2cm) node[tnode] (v) [label=right:$v$] {\small{$5^-$}}
-- ++(-120:1.2cm) node[whitenode] (x) {};

\draw (0.3,-0.1) node[anchor=text] (a) {$f$};
\draw (x) edge node  {} (u);

\draw (a) edge [post,bend left]  node [label=90:$1$] {} (v);
\end{tikzpicture}
\label{fig:r1}
}
\qquad
\subfloat[][$R_2$]{
\centering
\begin{tikzpicture}[scale=0.95]
\tikzstyle{whitenode}=[draw,circle,fill=white,minimum size=8pt,inner sep=0pt]
\tikzstyle{blacknode}=[draw,circle,fill=black,minimum size=6pt,inner sep=0pt]
\tikzstyle{tnode}=[draw,ellipse,fill=white,minimum size=8pt,inner sep=0pt]
\tikzstyle{texte} =[fill=white, text=black] 

\draw (0,0) node[whitenode] (u) {}
-- ++(0:1cm) node[whitenode] (w) {}
-- ++(-90+36:1cm) node[tnode] (v) [label=right:$v$] {\small{$5^-$}}
-- ++(-90-36:1cm) node[whitenode] (x) {}
-- ++(180:1cm) node[whitenode] (y) {};

\draw (0.3,-0.9) node[anchor=text] (a) {$f$};
\draw (y) edge [dashed, bend left] node  {} (u);

\draw (a) edge [post,bend left=20]  node [label=90:$2$] {} (v);
\end{tikzpicture}
\label{fig:r2}
}\\
\subfloat[][$R_3$]{
\centering
\begin{tikzpicture}[scale=1.35]
\tikzstyle{whitenode}=[draw,circle,fill=white,minimum size=8pt,inner sep=0pt]
\tikzstyle{blacknode}=[draw,circle,fill=black,minimum size=6pt,inner sep=0pt]
\tikzstyle{tnode}=[draw,ellipse,fill=white,minimum size=8pt,inner sep=0pt]
\tikzstyle{texte} =[fill=white, text=black]
\draw (0,0) node[tnode] (u) [label=left:$u$] {\small{$7^+$}}
-- ++(0:1cm) node[tnode] (v) [label=right:$v$ weak] {\small{$3$}};




\draw (u) edge [post,bend left=60]  node [label=90:$1$] {} (v);
\end{tikzpicture}
\label{fig:r3}
}
\qquad
\subfloat[][$R_4$]{
\centering
\begin{tikzpicture}[scale=1.25]
\tikzstyle{whitenode}=[draw,circle,fill=white,minimum size=8pt,inner sep=0pt]
\tikzstyle{blacknode}=[draw,circle,fill=black,minimum size=6pt,inner sep=0pt]
\tikzstyle{tnode}=[draw,ellipse,fill=white,minimum size=8pt,inner sep=0pt]
\tikzstyle{texte} =[fill=white, text=black]

\draw (0,0) node[tnode] (u) [label=left:$u$] {\small{$7^+$}}
-- ++(0:1cm) node[tnode] (v) [label=right:$v$ weak] {\small{$3$}};


\draw (u) edge [post,bend left=60]  node [label=90:$\frac{1}{2}$] {} (v);
\end{tikzpicture}
\label{fig:r4}
}
\qquad
\subfloat[][$R_5$]{
\centering
\begin{tikzpicture}[scale=1.35]
\tikzstyle{whitenode}=[draw,circle,fill=white,minimum size=8pt,inner sep=0pt]
\tikzstyle{blacknode}=[draw,circle,fill=black,minimum size=6pt,inner sep=0pt]
\tikzstyle{tnode}=[draw,ellipse,fill=white,minimum size=8pt,inner sep=0pt]
\tikzstyle{texte} =[fill=white, text=black]
\draw (0,0) node[tnode] (u) [label=left:$u$] {\small{$7^+$}}
-- ++(0:1cm) node[tnode] (v) [label=right:$v$ weak] {\small{$4$}};


\draw (u) edge [post,bend left=60]  node [label=90:$\frac{1}{2}$] {} (v);
\end{tikzpicture}
\label{fig:r5}
}\\
\subfloat[][$R_6$]{
\centering
\begin{tikzpicture}[scale=0.95]
\tikzstyle{whitenode}=[draw,circle,fill=white,minimum size=8pt,inner sep=0pt]
\tikzstyle{blacknode}=[draw,circle,fill=black,minimum size=6pt,inner sep=0pt]
\tikzstyle{invisible}=[draw=white,circle,fill=white,minimum size=6pt,inner sep=0pt]
\tikzstyle{tnode}=[draw,ellipse,fill=white,minimum size=8pt,inner sep=0pt]
\tikzstyle{texte} =[fill=white, text=black]
\draw (0,0) node[tnode] (u) [label=left:$u$] {\small{$8$}}
-- ++(0:1cm) node[tnode] (v) [label=right:$v$ $E_2$] {};

\draw (u) edge [post,bend left=60]  node [label=90:$\frac{1}{2}$] {} (v);
\end{tikzpicture}
\label{fig:r6}
}
\qquad
\subfloat[][$R_7$]{
\centering
\begin{tikzpicture}[scale=0.95]
\tikzstyle{whitenode}=[draw,circle,fill=white,minimum size=8pt,inner sep=0pt]
\tikzstyle{blacknode}=[draw,circle,fill=black,minimum size=6pt,inner sep=0pt]
\tikzstyle{tnode}=[draw,ellipse,fill=white,minimum size=8pt,inner sep=0pt]
\tikzstyle{texte} =[fill=white, text=black]
\tikzstyle{invisible}=[draw=white,circle,fill=white,minimum size=8pt,inner sep=0pt]

\draw (0,0) node[tnode] (u) [label=left:$u$] {\small{$8$}}
-- ++(0:1cm) node[tnode] (v) [label=right:$v$ $E_3$] {};

\draw (u) edge [post,bend left=60]  node [label=90:$\frac{1}{3}$] {} (v);
\end{tikzpicture}
\label{fig:r7}
}
\qquad
\subfloat[][$R_8$]{
\centering
\begin{tikzpicture}[scale=0.95]
\tikzstyle{whitenode}=[draw,circle,fill=white,minimum size=8pt,inner sep=0pt]
\tikzstyle{blacknode}=[draw,circle,fill=black,minimum size=6pt,inner sep=0pt]
\tikzstyle{tnode}=[draw,ellipse,fill=white,minimum size=8pt,inner sep=0pt]
\tikzstyle{texte} =[fill=white, text=black]
\tikzstyle{invisible}=[draw=white,circle,fill=white,minimum size=8pt,inner sep=0pt]
\draw (0,0) node[tnode] (u) [label=left:$u$] {\small{$8$}}
-- ++(0:1cm) node[tnode] (v) [label=right:$v$ $E_4$] {};

\draw (u) edge [post,bend left=60]  node [label=90:$\frac{1}{4}$] {} (v);
\end{tikzpicture}
\label{fig:r8}
}\\
\subfloat[][$R_9$]{
\centering
\begin{tikzpicture}[scale=0.95]
\tikzstyle{whitenode}=[draw,circle,fill=white,minimum size=8pt,inner sep=0pt]
\tikzstyle{blacknode}=[draw,circle,fill=black,minimum size=6pt,inner sep=0pt]
\tikzstyle{tnode}=[draw,ellipse,fill=white,minimum size=8pt,inner sep=0pt]
\tikzstyle{texte} =[fill=white, text=black]
\tikzstyle{invisible}=[draw=white,circle,fill=white,minimum size=8pt,inner sep=0pt]

\draw (0,0) node[tnode] (u) [label=left:$u$] {\small{$7$}}
-- ++(0:1cm) node[tnode] (v) [label=right:$v$ $S_2$] {};

\draw (u) edge [post,bend left=60]  node [label=90:$\frac{1}{2}$] {} (v);
\end{tikzpicture}
\label{fig:r9}
}
\qquad
\subfloat[][$R_{10}$]{
\begin{tikzpicture}[scale=0.95]
\tikzstyle{whitenode}=[draw,circle,fill=white,minimum size=8pt,inner sep=0pt]
\tikzstyle{blacknode}=[draw,circle,fill=black,minimum size=6pt,inner sep=0pt]
\tikzstyle{invisible}=[draw=white,circle,fill=white,minimum size=6pt,inner sep=0pt]
\tikzstyle{tnode}=[draw,ellipse,fill=white,minimum size=8pt,inner sep=0pt]
\tikzstyle{texte} =[fill=white, text=black]

\draw (0,0) node[tnode] (u) [label=left:$u$] {\small{$7$}}
-- ++(0:1cm) node[tnode] (v) [label=right:$v$ $S_3$] {};

\draw (u) edge [post,bend left=60]  node [label=90:$\frac{1}{3}$] {} (v);
\end{tikzpicture}
\label{fig:r10}
}
\qquad
\subfloat[][$R_{11}$]{
\begin{tikzpicture}[scale=0.95]
\tikzstyle{whitenode}=[draw,circle,fill=white,minimum size=8pt,inner sep=0pt]
\tikzstyle{blacknode}=[draw,circle,fill=black,minimum size=6pt,inner sep=0pt]
\tikzstyle{tnode}=[draw,ellipse,fill=white,minimum size=8pt,inner sep=0pt]
\tikzstyle{texte} =[fill=white, text=black]
\tikzstyle{invisible}=[draw=white,circle,fill=white,minimum size=8pt,inner sep=0pt]
\draw (0,0) node[tnode] (u) [label=left:$u$] {\small{$7$}}
-- ++(0:1cm) node[tnode] (v) [label=right:$v$ $S_4$] {};

\draw (u) edge [post,bend left=60]  node [label=90:$\frac{1}{4}$] {} (v);
\end{tikzpicture}
\label{fig:r11}
}
\caption{Discharging rules.}
\label{fig:dis}
\end{figure}
\captionsetup[subfloat]{labelformat=parens}

Note that according to these rules, only vertices of degree at most $5$ receive weight, and only faces of degree at least $4$ and vertices of degree at least $7$ give weight. Note that the notation $E_i$ and $S_i$ corresponds to the fact that a vertex $u$ gives a weight of $\frac{1}{i}$ to every $E_i$- or $S_i$-neighbor.

\begin{lemma}\label{lem:rules}
A planar graph $G$ with $\Delta(G) \leq 8$ that does not contain Configurations \textbf{($C_1$)} to \textbf{($C_{11}$)} is a stable set.
\end{lemma}

\begin{proof}
We can assume without loss of generality that $G$ is connected (if it is not, we simply consider a connected component of $G$, as it verifies the same hypothesis). Assume by contradiction that $G$ is not a single vertex. Thus $G$ is connected and contains at least one edge. According to Configuration \textbf{($C_1$)}, every vertex $x$ of $G$ verifies $d(x) \geq 3$. We consider a planar embedding of $G$.\\
We attribute to each vertex $u$ a weight of $d(u)-6$, and to each face a weight of $2d(f)-6$. We apply discharging rules $R_1$, $R_2$, $\ldots$, $R_{11}$. We show that all the faces and vertices have a weight of at least $0$ in the end.\newline

Note that the degree of a face is the number of vertices on its boundary, while walking through a facial walk (i.e. some vertices are counted with multiplicity). The discharging rules on the faces also apply with multiplicity: $R_1$ and $R_2$ apply to each vertex of degree at most $5$ incident to $f$ as many times as it appears on the boundary of $f$.

Let $f$ be a face in $G$. By Configuration \textbf{$(C_1)$}, no two vertices of degree at most $5$ are adjacent. Thus $f$ is incident to at most $\lfloor \frac{d(f)}{2} \rfloor$ vertices of degree $\leq 5$. We consider four cases depending on $d(f)$.
\begin{enumerate}
\item \textit{$d(f)=3$}. Then $f$ has an initial weight of $0$ and gives nothing, so it has a final weight of at least $0$.
\item \textit{$d(f)=4$}. Face $f$ is incident to at most $2$ vertices of degree $\leq 5$. So $f$ has an initial weight of $2$ and gives at most two times $1$ according to $R_1$. Thus $f$ has a final weight of at least $2- 2 \times 1 \geq 0$.
\item \textit{$d(f)=5$}. Face $f$ is incident to at most $2$ vertices of degree $\leq 5$. So $f$ has an initial weight of $4$ and gives at most two times $2$ according to $R_2$. Thus $f$ has a final weight of at least $4- 2 \times 2 \geq 0$.
\item \textit{$d(f)\geq 6$}. Face $f$ is incident to at most $\lfloor \frac{d(f)}{2} \rfloor \leq \frac{d(f)}{2}$ vertices of degree $\leq 5$. So $f$ has an initial weight of $2\times d(f)-6$ and gives at most $\frac{d(f)}{2}$ times $2$ according to $R_2$. Thus $f$ has a final weight of at least $2 \times d(f) - 6 - 2 \times \frac{d(f)}{2} = d(f) - 6 \geq 0$.
\end{enumerate}

So all the faces have a final weight of at least $0$ after application of the discharging rules. Let us now prove that the same holds for the vertices.\\

Let $x$ be a vertex of $G$. We consider different cases corresponding to the value of $d(x)$.

\begin{enumerate}

\item $d(x)=3$. Vertex $x$ has an initial weight of $-3$. We show that it receives at least $3$, thus has a non-negative final weight. By Configuration \textbf{$(C_1)$}, the three neighbors of $x$ are of degree $8$. We consider four cases depending on the degrees of the three faces $f_1, f_2$ and $f_3$ incident to $x$. We assume $d(f_1) \geq d(f_2) \geq d(f_3)$. Let $u_1, u_2$ and $u_3$ be the three neighbors of $u$, where for every $i \in \{1,2,3\}$, the edge $(x,u_i)$ belongs to $f_{i-1}$ and $f_i$ (subscripts taken modulo $3$).

\begin{enumerate}
\item \textit{$d(f_1) \geq 5$ and $d(f_2) \geq 4$}.\\So $x$ receives $2$ from $f_1$ by $R_2$, and at least $1$ from $f_2$ by $R_1$ or $R_2$.
\item \textit{$d(f_1)=d(f_2)=d(f_3)=4$}.\\So $x$ receives $1$ from each $f_i$ by $R_1$.
\item \textit{$d(f_1)=d(f_2)=4$ and $d(f_3)=3$}.\\So $x$ receives $1$ from both $f_1$ and $f_2$ by $R_1$. Besides, $x$ is a semi-weak neighbor of $u_1$ and $u_3$, so $x$ receives $\frac{1}{2}$ from $u_1$ and $u_2$ by $R_4$.
\item \textit{$d(f_1)\geq 5$ and $d(f_2)=d(f_3)=3$}.\\So $x$ receives $2$ from $f_1$ by $R_2$. Vertex $x$ is a weak neighbor of $u_3$, so $x$ receives $1$ from $u_3$ by $R_3$.
\item \textit{$d(f_1)=4$, and $d(f_2)=d(f_3)=3$}.\\So $x$ receives $1$ from $f_1$ by $R_1$. Besides, $x$ is a weak neighbor of $u_3$ and a semi-weak neighbor of $u_1$ and $u_2$, so $x$ receives $1$ from $u_3$ by $R_3$, and $\frac{1}{2}$ from both $u_1$ and $u_2$ by $R_4$.
\item \textit{$d(f_1)=d(f_2)=d(f_3)=3$}.\\Then $x$ is a weak neighbor of $u_1$, $u_2$ and $u_3$, so $x$ receives $1$ from $u_1$, $u_2$ and $u_3$ by $R_3$.
\end{enumerate}

\item $d(x)=4$. Vertex $x$ has an initial weight of $-2$. We show that it receives at least $2$, thus has a non-negative final weight. By Configuration \textbf{$(C_1)$}, the four neighbors $u_1,u_2,u_3$ and $u_4$ of $x$ are of degree at least $7$. We consider three cases depending on how many triangles are incident to $x$.
\begin{enumerate}
\item \textit{Vertex $x$ is incident to at most $2$ triangles}.\\
Then $x$ is incident to at least two faces $f_1$ and $f_2$ with $d(f_1),d(f_2) \geq 4$. So $x$ receives at least $1$ from both $f_1$ and $f_2$ by $R_1$ or $R_2$.
\item \textit{Vertex $x$ is incident to exactly $3$ triangles $(x,u_1,u_2), (x,u_2,u_3)$ and $(x,u_3,u_4)$}.\\
Then $x$ is incident to a face $f_1$ with $d(f_1) \geq 4$. So $x$ receives at least $1$ from $f_1$ by $R_1$ or $R_2$. Besides, $x$ is a weak neighbor of $u_2$ and $u_3$, so $x$ receives $\frac{1}{2}$ from both $u_2$ and $u_3$ by $R_5$.
\item \textit{Vertex $x$ is incident to $4$ triangles}.\\
Then $x$ is a weak neighbor of $u_1, u_2, u_3$ and $u_4$, so $x$ receives $\frac{1}{2}$ from $u_1, u_2, u_3$ and $u_4$ by $R_5$.
\end{enumerate}

\item $d(x)=5$. Vertex $x$ has an initial weight of $-1$. We show that it receives at least $1$, thus has a non-negative final weight. By Configuration $(C_1)$, the five consecutive neighbors $u_1,u_2,u_3, u_4$ and $u_5$ of $x$ are of degree at least $6$.

In the case where $x$ is incident to a face $f$ with $d(f)\geq 4$, vertex $x$ receives at least $1$ from $f$ by $R_1$ or $R_2$. So we can assume that $x$ is incident to five triangles $(x,u_1,u_2)$, $(x,u_2,u_3)$, $(x,u_3,u_4)$, $(x,u_4,u_5)$ and $(x,u_5,u_1)$. We consider four cases depending on the number of vertices of degree $6$ incident to $x$.
\begin{enumerate}
\item \textit{Vertex $x$ has at least three neighbors of degree $6$.}\\
By Configuration \textbf{$(C_{11})$}, they cannot appear consecutively around $x$, so they are exactly three. Without loss of generality, we assume $d(u_1)=d(u_2)=d(u_4)=6$, hence $d(u_3), d(u_5) \geq 7$. Then $x$ is an $E_2$- or $S_2$-neighbor of $u_3$ and $u_5$, so it receives $\frac{1}{2}$ from both $u_3$ and $u_5$ by $R_6$ or $R_9$.
\item \textit{Vertex $x$ has exactly two neighbors of degree $6$.}\\
We consider two cases depending on whether these vertices of degree $6$ appear consecutively around $x$.
\begin{enumerate}
\item \textit{Vertex $x$ has two consecutive neighbors of degree $6$}.\\
We can assume w.l.o.g. that $d(u_1)=d(u_2)=6$, and that $d(u_3) \geq d(u_5)$. We consider three cases depending on $d(u_3)$ and $d(u_5)$.
\begin{enumerate}
\item \textit{$d(u_3)=d(u_5)=8$}.\\
Then $x$ is an $E_2$-neighbor of $u_3$ and $u_5$, so $x$ receives $\frac{1}{2}$ from both $u_3$ and $u_5$ by $R_6$.
\item \textit{$d(u_3)=8$, $d(u_5)=7$}.\\
Then $x$ is an $E_2$-neighbor of $u_3$, an $S_3$- or $S_4$-neighbor of $u_5$ (depending on the degree of $u_4$), and an $S_4$- or $E_3$-neighbor of $u_4$, so $x$ receives $\frac{1}{2}$ from $u_3$ by $R_6$, and at least $\frac{1}{4}$ from both $u_4$ and $u_5$ by $R_7$, $R_{10}$ or $R_{11}$.
\item \textit{$d(u_3)=d(u_5)=7$}.\\
Then $x$ is an $S_3$-neighbor of $u_3$ and $u_5$, and an $S_3$- or $E_3$-neighbor of $u_4$, so $x$ receives $\frac{1}{3}$ from $u_3$, $u_4$ and $u_5$ by $R_7$ or $R_{10}$.
\end{enumerate}
\item \textit{Vertex $x$ has no two consecutive neighbors of degree $6$}.\\
We can assume without loss of generality that $d(u_1)=d(u_4)=6$ and that $d(u_2) \geq d(u_3)$.
We consider two cases depending on $d(u_3)$.
\begin{enumerate}
\item \textit{$d(u_3)=8$}.\\
Then $d(u_2)=8$. Vertex $x$ is an $E_3$- or $S_2$-neighbor of $u_5$, and an $E_3$-neighbor of $u_2$ and $u_3$, so $x$ receives at least $\frac{1}{3}$ from $u_2$, $u_3$ and $u_5$ by $R_7$ or $R_9$.
\item \textit{$d(u_3)=7$}.\\
Then $x$ is an $E_2$- or $S_2$-neighbor of $u_5$, and an $S_3$-, $S_4$- or $E_3$-neighbor of $u_2$ and $u_3$, so $x$ receives $\frac{1}{2}$ from $u_5$ by $R_6$ or $R_9$, and at least $\frac{1}{4}$ from $u_2$ and $u_3$ by $R_7$, $R_{10}$ or $R_{11}$.
\end{enumerate}
\end{enumerate}
\item \textit{Vertex $x$ has exactly one neighbor of degree $6$.}\\
We can assume without loss of generality that $d(u_1)=6$, and $d(u_2)\geq d(u_5)$ or $d(u_3) \geq d(u_4)$ if $d(u_2)=d(u_5)$.
We consider three cases depending on $d(u_5)$ and $d(u_3)$.
\begin{enumerate}
\item \textit{$d(u_5)=8$ and $d(u_3)=d(u_4)$}.\\
Then $x$ is an $E_3$-neighbor of $u_2$ and $u_5$, so it receives $\frac{1}{3}$ from both by $R_7$. Besides, since $d(u_3)=d(u_4)$, vertex $x$ is an $S_4$- or $E_4$-neighbor of $u_3$ and $u_4$, so it receives $\frac{1}{4}$ from both by $R_8$ or $R_{11}$.
\item \textit{$d(u_5)=8$ and $d(u_3)\neq d(u_4)$}.\\
Then $d(u_2)=d(u_3)=8$ and $d(u_4)=7$. Vertex $x$ is an $E_3$-neighbor of $u_2$, $u_3$ and $u_5$, so it receives $\frac{1}{3}$ from each by $R_7$.
\item \textit{$d(u_5)=7$}.\\
Then vertex $x$ is an $E_3$-, $E_4$- or $S_4$-neighbor of every $u_i$ for $i \in \{2,3,4,5\}$, so it receives at least $\frac{1}{4}$ from each by $R_7$, $R_8$ or $R_{11}$.
\end{enumerate}
\item \textit{Vertex $x$ has no neighbor of degree $6$.}\\
We consider three cases depending on the degrees of the $u_i$'s.
\begin{enumerate}
\item \textit{Vertex $x$ has at least $4$ neighbors of degree $8$.}\\
Then $x$ is an $E_3$- or $E_4$-neighbor of each of them, so it receives at least $\frac{1}{4}$ from each by $R_7$ or $R_8$.
\item \textit{Vertex $x$ has two consecutive neighbors of degree $7$.}\\
We consider w.l.o.g. that $d(u_1)=d(u_2)=7$. Then $x$ is an $S_4$-neighbor of $u_1$ and $u_2$, so it receives at least $\frac{1}{4}$ from each by $R_{11}$. Vertex $x$ is also an $S_4$- or $E_3$-neighbor of $u_3$ and $u_5$, so it receives at least $\frac{1}{4}$ from each by $R_7$ or $R_{11}$.
\item \textit{Vertex $x$ has at most $3$ neighbors of degree $8$, and has no two consecutive neighbors of degree $7$.}\\
Since $x$ is only adjacent to vertices of degree $7$ or $8$, we consider w.l.o.g. that $d(u_1)=d(u_3)=7$, and $d(u_2)=d(u_4)=d(u_5)=8$. Then $x$ is an $E_3$-neighbor of $u_2$, $u_4$ and $u_5$, so it receives $\frac{1}{3}$ from each by $R_7$.
\end{enumerate}
\end{enumerate}
\item $d(x)=6$. Vertex $x$ has an initial weight of $0$, gives nothing away, and has a final weight of at least $0$.
\item $d(x)=7$. Vertex $x$ has an initial weight of $1$. We show that it gives at most  $1$, thus has a non-negative final weight. By Configuration \textbf{$(C_1)$}, the neighbors of $x$ have degree at least $4$, and $x$ has at most $3$ weak neighbors of degree at most $5$. We consider four cases depending on the weak neighbors of $x$.

\begin{enumerate}
\item \textit{Vertex $x$ has an $S_2$-neighbor $v$}.\\
Let $v, w_1, w_2, w_3, w_4, w_5$ and $w_6$ be the consecutive neighbors of $x$. By definition of an $S_2$-neighbor, $d(w_1)=d(w_6)=6$. By Configuration \textbf{$(C_8)$}, if $w_2$ (resp. $w_5$) is a weak neighbor of $x$, then $d(w_2)>5$ (resp. $d(w_5)>5$). Assume w.l.o.g. that $d(w_3) \geq d(w_4)$. Then by Configuration \textbf{$(C_1)$}, if $w_3$ and $w_4$ are adjacent then $d(w_3)>5$. Thus $x$ has at most two weak neighbors of degree at most $5$: $v$ and possibly $w_4$. Besides, $d(v), d(w_4) > 3$. By Rules $R_5$, $R_9$, $R_{10}$ and $R_{11}$, vertex $x$ gives at most $\frac{1}{2}$ to each.
\item \textit{Vertex $x$ has at least two weak neighbors of degree $4$}.\\ 
By Configuration \textbf{$(C_9)$}, $x$ is adjacent to no other weak neighbor of degree $4$, and no $S_2$, $S_3$ or $S_4$-neighbor. Thus $x$ gives $\frac{1}{2}$ to each of the two weak neighbors of degree $4$ by $R_5$.
\item \textit{Vertex $x$ has exactly one weak neighbor $v$ of degree $4$ and no $S_2$-neighbor}.\\
If $x$ has an $S_3$-neighbor $v_2$, then by Configuration \textbf{$(C_{10})$}, it has no other neighbor of degree at most $5$. Thus $x$ gives $\frac{1}{2}$ to $v$ by $R_5$, $\frac{1}{3}$ to $v_2$ by $R_{10}$.\\
If $x$ has no $S_3$-neighbor, then $x$ has at most two other weak neighbors $v_1$ and $v_2$ of degree at most $5$, which are of degree $5$ by assumption. So $x$ gives $\frac{1}{2}$ to $v$ by $R_5$, $\frac{1}{4}$ to $v_1$ and $v_2$ by $R_{11}$.
\item \textit{Vertex $x$ has no weak neighbor of degree $4$, and no $S_2$-neighbor}.\\
Vertex $x$ has at most three weak neighbors $v_1$, $v_2$ and $v_3$ of degree at most $5$, which are of degree $5$ by assumption. So $x$ gives at most $\frac{1}{3}$ to each by $R_{10}$ or $R_{11}$.
\end{enumerate}

\item $d(x)=8$. Vertex $x$ has an initial weight of $2$. We show that it gives at most $2$, thus has a non-negative final weight. By Configurations \textbf{$(C_1)$} and \textbf{$(C_2)$}, vertex $x$ has at most $4$ neighbors that are either semi-weak with degree $3$ or weak with degree at most $5$. We consider eight cases depending on the neighborhood of $x$.
\begin{enumerate}
\item \textit{Vertex $x$ has at least two weak neighbors $v_1$ and $v_2$ of degree $3$}.\\
Then by Configuration \textbf{$(C_3)$}, vertex $x$ has exactly two neighbors of degree at most $5$. Thus $x$ gives $1$ to $v_1$ and $v_2$ by $R_{3}$.
\item \textit{Vertex $x$ has exactly one weak neighbor $v_1$ of degree $3$, and at least one semi-weak neighbor $v_2$ of degree $3$}.\\
Then by Configuration \textbf{$(C_4)$}, vertex $x$ has at most one other neighbor $v_3$ of degree at most $5$. By assumption, vertex $v_3$ is not a weak neighbor of $x$ of degree $3$, so $x$ gives at most $\frac{1}{2}$ to $v_3$ by $R_4$, $R_5$, $R_6$, $R_7$ or $R_8$. Vertex $x$ gives $1$ to $v_1$ by $R_3$, and $\frac{1}{2}$ to $v_2$ by $R_4$.
\item \textit{Vertex $x$ has exactly one weak neighbor $v_1$ of degree $3$, no semi-weak neighbor of degree $3$, and at least two weak neighbors $v_2$ and $v_3$ of degree $4$}.\\
Then, by Configuration \textbf{$(C_5)$}, vertex $x$ has no other weak neighbor of degree at most $5$. By assumption, it has no semi-weak neighbor of degree $3$. So $x$ gives $1$ to $v_1$ by $R_3$, $\frac{1}{2}$ to $v_2$ and $v_3$ by $R_5$.
\item \textit{Vertex $x$ has exactly one weak neighbor $v_1$ of degree $3$, no semi-weak neighbor of degree $3$, exactly one weak neighbor $v_2$ of degree $4$, and at least one $E_2$- or $E_3$-neighbor $v_3$}.\\
By definition of $E_2$- and $E_3$-neighbor, vertices $x$ and $v_3$ have a common neighbor $v_4$ of degree at most $7$, which by Configuration $(C_1)$ has degree $6$ or $7$. Then, by Configuration \textbf{$(C_6)$}, vertex $x$ has no other neighbor of degree at most $5$. So $x$ gives $1$ to $v_1$ by $R_3$, $\frac{1}{2}$ to $v_2$ by $R_5$, at most $\frac{1}{2}$ to $v_3$ by $R_6$ or $R_7$.
\item \textit{Vertex $x$ has exactly one weak neighbor $v_1$ of degree $3$, no semi-weak neighbor of degree $3$, exactly one weak neighbor $v_2$ of degree $4$, and no $E_2$- or $E_3$-neighbor}.\\
Then $x$ has at most two other weak neighbors $v_3$ and $v_4$ of degree at most $5$, which are by assumption $E_4$-neighbors.
So $x$ gives $1$ to $v_1$ by $R_3$, $\frac{1}{2}$ to $v_2$ by $R_5$, $\frac{1}{4}$ to $v_3$ and $v_4$ by $R_8$.
\item \textit{Vertex $x$ has exactly one weak neighbor $v_1$ of degree $3$, no semi-weak neighbor of degree $3$, no weak neighbor $v_2$ of degree $4$, and at least an $E_2$-neighbor $v_2$}.\\
Then by Configuration \textbf{$(C_7)$}, vertex $x$ has at most one other weak neighbor $v_3$ of degree at most $5$, which is by assumption of degree $5$.
So $x$ gives $1$ to $v_1$ by $R_3$, at most $\frac{1}{2}$ to $v_2$ and $v_3$ by $R_6$, $R_7$ or $R_8$.
\item \textit{Vertex $x$ has exactly one weak neighbor $v_1$ of degree $3$, no weak neighbor $v_2$ of degree $4$, no semi-weak neighbor of degree $3$, and no $E_2$-neighbor}.\\
Then $x$ has at most three other weak neighbors $v_2$, $v_3$ and $v_4$ of degree at most $5$, which are by assumption of degree $5$. Vertex $x$ has no $E_2$-neighbor, so they are $E_3$ or $E_4$-neighbors of $x$.
So $x$ gives $1$ to $v_1$ by $R_3$, at most $\frac{1}{3}$ to $v_2$, $v_3$ and $v_4$ by $R_7$ or $R_8$.
\item \textit{Vertex $x$ has no weak neighbor of degree $3$}.\\
Then $x$ has at most four neighbors $v_1$, $v_2$, $v_3$ and $v_4$ of degree at most $5$ that are either weak with degree at least $4$ or semi-weak with degree $3$. So $x$ gives at most $\frac{1}{2}$ to each by $R_4$, $R_5$, $R_6$, $R_7$ or $R_8$.
\end{enumerate}
\end{enumerate}

Consequently, after application of the discharging rules, every vertex and every face of $G$ has a non-negative weight, $6|E|-6|V|-6|F|=(2|E|-6|V|)+(4|E|-6|F|)=\sum_{v \in V} (d(v)-6)+ \sum_{f \in F} (2d(f)-6) \geq 0$, a contradiction to Euler's Formula.

\end{proof}

\section{Proof of Theorem~\ref{thm:main}}\label{sect:ccl}

Let $G$ be a minimal planar graph with $\Delta(G) \leq 8$ such that $G$ is not $9$-edge-choosable. By Lemma~\ref{lem:config}, graph $G$ cannot contain \textbf{($C_1$)} to \textbf{($C_{11}$)}. Lemma~\ref{lem:rules} implies that $G$ is a stable set, thus $9$-edge-choosable, a contradiction.
\hfill $\Box$ \vspace{1em}

\section{Conclusion}
The key idea in the proof lies in some recoloring arguments using directed graphs (see Claims~\ref{claim:C3}, \ref{claim:C4} and \ref{claim:C6} for occurrences of it in the proof). It allowed us to deal with configurations that would not yield under usual techniques, and thus to improve Theorem~\ref{thm:D9}. Though this simple argument does not seem to be enough to prove Conjecture~\ref{conj:LCCw} for $\Delta=7$, it might be interesting to try to improve similarly Theorem~\ref{thm:D12}.  

Note that the proof could easily be adapted to prove that planar graphs with $\Delta \geq 8$ are $(\Delta+1)$-edge-choosable. This would however be of little interest considering the simple proof for $\Delta \geq 9$ presented in~\cite{ch10}.

Conjecture~\ref{conj:LCCw} remains open for $\Delta=5,6$ and $7$. It might be interesting to weaken the conjecture and ask whether all planar graphs are $(\Delta+2)$-edge-choosable. This is true for planar graphs with $\Delta \geq 7$ by Theorems~\ref{thm:D9} and \ref{thm:main}. What about planar graphs with $\Delta=6$?


\bibliographystyle{plain}

\end{document}